\documentclass{uwstat518}

\usepackage{graphicx}
\usepackage{color}
\usepackage{array}
\usepackage{multirow}
\usepackage{mathtools}

\usepackage{float}
\usepackage{bbm}
\usepackage{amsmath}
\usepackage{amssymb}
\usepackage{ulem}
\usepackage{subfig}
\usepackage{wrapfig}
\usepackage{url}

\usepackage{natbib}
\DeclareMathOperator*{\argmax}{argmax}

\newcommand{\mb}{\mathbf}

\newcommand{\indfun}[1]{\ensuremath{\mb{1}_{\{#1\}}}}

\newtheorem{theorem}{Theorem}[section]

\newenvironment{proof}[1][Proof]{\begin{trivlist}
\item[\hskip \labelsep {\bfseries #1}]}{\end{trivlist}}

\newcommand{\qed}{\nobreak \ifvmode \relax \else
      \ifdim\lastskip<1.5em \hskip-\lastskip
      \hskip1.5em plus0em minus0.5em \fi \nobreak
      \vrule height0.75em width0.5em depth0.25em\fi}

\begin{document}
\begin{center}
\large{Likelihood-Based Inference for Discretely Observed Birth-Death-Shift Processes, with
Applications to Evolution of Mobile Genetic Elements}

\normalsize
By Jason Xu$^\ast$, Peter Guttorp$^\ast$, Midori Kato-Maeda$^\dagger$, Vladimir Minin$^{\ast, \ddagger}$

\small
$^\ast$Department of Statistics, University of Washington\\
$^\ddagger$Department of Biology, University of Washington  \\
$^\dagger$School of Medicine, University of California, San Francisco \\

\normalsize
\end{center}


\begin{abstract}
Continuous-time birth-death-shift (BDS) processes are frequently used in stochastic modeling, with many applications in ecology and epidemiology. In particular, such processes can model evolutionary dynamics of transposable elements --- important genetic markers in molecular
epidemiology. Estimation of the effects of individual covariates on the birth, death, and shift rates of the process can be accomplished by analyzing patient data, but inferring these rates in a discretely and unevenly observed setting presents computational challenges. We propose
a multi-type branching process approximation to BDS processes and develop a corresponding expectation maximization algorithm, where we use spectral techniques to reduce calculation of expected sufficient statistics to low dimensional integration. These techniques yield an efficient and robust optimization routine for inferring the rates of the BDS process, and apply more broadly to multi-type branching processes where rates can depend on many covariates. 
After rigorously testing our methodology in simulation studies, we apply our method to study intrapatient time evolution of IS\textit{6110}
transposable element, a  genetic marker frequently used during estimation of epidemiological clusters of \textit{Mycobacterium tuberculosis} infections.
 \end{abstract}

\section{Introduction}
Continuous-time branching processes are widely used in stochastic modeling of population dynamics. Originally introduced as a mathematical model for the survival of family surnames, the tools from branching process theory have since found a breadth of applications including biology, genetics, epidemiology, quantum optics, and nuclear fission \citep{renshaw2011}. One of the most widely used classes of branching processes are birth-death (BD) processes, a simple yet flexible model for single-species population dynamics. The popularity of BD processes is in part attributable to their well-understood mathematical properties.
To accurately model behavior in many applications, however,  it is often necessary to consider systems with more than one species --- bivariate or other multi-type processes are commonly used to model phenomena such as competition, predation, or infection \citep{neyman1956, alsmeyer1993, renshaw2011}. Multi-type branching processes form one class of models that can accommodate populations with multiple types, but these models pose considerable
computational challenges for statistical inference. Our work introduces new methods 
to overcome these challenges, enabling likelihood-based inference in partially observed, multi-type branching processes. 
\par
Many statistically relevant quantities are available in closed form for the linear, homogeneous BD process and several of its variants, including transition probabilities, stationary distributions, and moments \citep{bailey1964, keiding1975}. Further, analytical expressions of transition probabilities as series of orthogonal polynomials are known for general, nonlinear BD processes \citep{karlin1958} and can be conveniently computed numerically \citep{murphy1975, crawford2012}.  The ability to compute finite-time transition probabilities enables likelihood-based inference for discretely observed or partially observed BD processes, since the observed likelihood is a function of these transition probabilities. Evaluating this likelihood is necessary in maximum likelihood estimation as well as in many Bayesian inferential procedures. Recent work by \citet{doss2013} and \citet{Crawford2013} introduces techniques to additionally compute conditional moments of BD sufficient statistics for linear and general birth-death-immigration processes
, enabling calculation of the expected complete-data likelihood necessary in an expectation-maximization (EM) algorithm \citep{dempster1977}.
\par
Unfortunately, methods to evaluate finite-time transition probabilities and conditional moments are not known in the multi-type setting, and generalizing the techniques available in the single-species case is nontrivial. Solutions to the Kolmogorov equations in multi-type settings are available only for several linear, closed systems such as the immigration-death-shift process, but simple modifications such as the presence of birth events significantly complicate analysis \citep{puri1968, renshaw2011}. Without these quantities, likelihood-based estimation is limited to simulation-based inference via Monte Carlo EM or MCMC \citep{golinelli2000, golinelli2006} and asymptotic approximations, such as moment-based estimating equations \citep{catlin2001}. However, these approaches have shortcomings. MCMC approaches require augmenting the state space by high-dimensional latent variables and become computationally prohibitive when the state space is large. Moment-based methods are statistically less efficient than likelihood-based approaches and thus often inappropriate for smaller datasets, requiring a large number of observations to produce meaningful standard errors and confidence intervals.
\par
In this paper, we extend the analysis of \citet{doss2013}, deriving previously unavailable numerical solutions to transition probabilities and conditional moments for discretely observed, multi-type branching processes.  
Modifying ideas introduced by \cite{kendall1948}, we simplify the systems of backward equations for several relevant generating functions, and then apply the spectral approach of \cite{lange1982} to extract expected sufficient statistics and transition probabilities. This enables us to evaluate the observed likelihood, as well as to reduce the challenging computation of expected complete-data log-likelihood necessary in an EM algorithm to efficient evaluation of expected sufficient statistics by low-dimensional integration. Our EM algorithm can be applied in settings where the data are assumed to be generated from independent, continuous-time multi-type branching processes, observed at discrete and possibly irregularly spaced time points, whose rates can be a function of many process-specific covariates.
Medical applications, for instance, commonly feature such panel data, where rates corresponding to the transmission of a disease or growth of a cell may depend on patient-specific characteristics. While similar methods have been explored for fitting continuous-time finite state-space Markov chains to panel data \citep{jackson2011, kalbfleisch1985, lange1995, lange2013}, our method allows for multivariate and potentially infinite state-space processes. 
\par
Though our methodology applies broadly to multi-type Galton Watson branching processes, we focus attention to estimating the rates of a birth-death-shift (BDS) process. The BDS process adds the possibility of shift events to the standard BD framework, and is useful for modeling systems that allow for elements to switch locations or types --- a shift is essentially a simultaneous birth and death. For example, in epidemiological applications, interaction between infected and susceptible populations can be captured as a shift event, involving a simultaneous increase and decrease in the respective populations. Spatial BDS processes have also been studied to improve Metropolis-Hastings algorithms for perfect sampling \citep{huber2012} relevant to a range of spatial statistical applications; see \cite{illian2008} for an overview. 
Our motivation stems from the BDS process proposed by \cite{rosenberg2003} to model evolution of transposons --- mobile genetic elements that 
can replicate, die, or shift locations along the genome. Specifically, \cite{rosenberg2003} study the within-host evolution of the IS\textit{6110} transposon in the \textit{Mycobacterium tuberculosis} genome via a BDS model. Accurately estimating the rates of these events is important in molecular epidemiology \citep{tanaka2001}. 
\par
\cite{rosenberg2003} infer the birth, death, and shift rates of the BDS model of IS\textit{6110} evolution from an ongoing database of \textit{M. tuberculosis} patients from San Francisco, but their rate estimates rely on approximate maximum likelihood estimators (MLEs) under a restrictive assumption that at most one event occurs per observation interval. This study was revisited in \citep{doss2013}, in which the authors derive an EM algorithm for inference in discretely observed birth-death-immigration (BDI) processes amenable to high-dimensional optimization. Although this approach more realistically allows multiple events to occur per interval, the methodology in \citep{doss2013} is limited to the single-species setting, effectively replacing the BDS model with a simpler model that ignores particle locations and shift events.
Our methodology extends the analysis of \cite{doss2013}, allowing for both the possibility of multiple events per observation interval as well as the consideration of shift events. We show that the dynamics of the BDS model can be captured in a two-type branching process framework in that transition probabilities of both processes are nearly identical. We then derive an EM algorithm for discretely observed multi-type branching processes, and rigorously assess its performance in several simulation studies. Finally, we revisit the San Francisco tuberculosis dataset, applying our algorithm to estimate rates of the IS\textit{6110} transposon as a function of relevant covariates.

\section{Methodology}


\subsection{Birth-death-shift model for transposable elements}
Our motivation stems from a birth-death-shift process proposed by  \citet{rosenberg2003} to model evolutionary dynamics of transposable elements or transposons --- genomic mobile sequence elements. Each transposon can (1) duplicate, with the new copy moving to a new genomic location; (2) shift to a different genomic position; or (3) be removed and lost from the genome, independently of all other transposons. These events occur at instantaneous rates proportional to the total transposon copy number at that time. Thus, transposons evolve according to a linear birth-death-shift (BDS) process in continuous time.

The process of transposon evolution within a host is observable by serially genotyping the organism of interest, e.g., \textit{Mycobacterium tuberculosis} as in \citet{rosenberg2003}. 
\textit {M. tuberculosis} genome typically has between 0 and 25 copies of the IS6110 element. The number and chromosomal position of the IS\textit{6110} element can be visualized using restriction fragment length polymorphism (RFLP). This technique entails restriction endonuclease digestion of the \textit{M. tuberculosis} DNA which is run in an agarose gel, southern blotting and probing with a peroxidase labeled IS\textit{6110} probe. Birth, death, and shift events are thus detectable via changes in the number and size of the bands where the IS6110 elements are located.


\begin{figure}[t]
\centering
\includegraphics[width = .38\paperwidth]{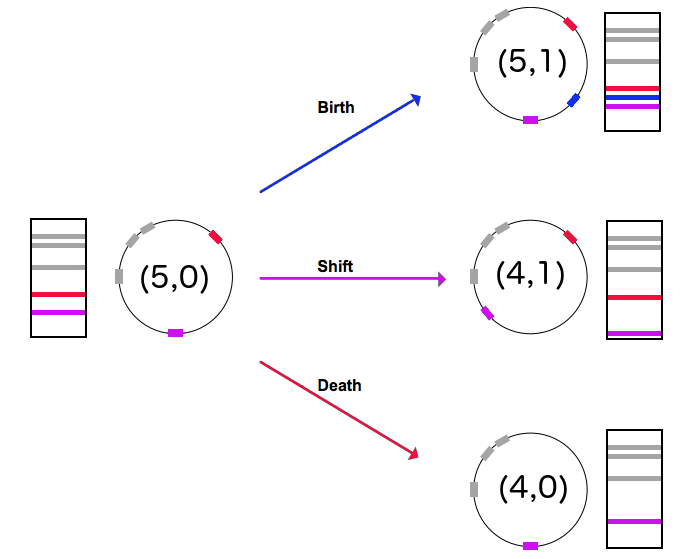}
\caption{\small Illustration of the three types of transposition---birth, death, shift---along a genome, represented by circles. Transposons, depicted by filled rectangles along the circles/genomes, correspond to observable gel bands, denoted by horizontal lines in the rectangles next to each circle diagram. Numbers within each circle represent each configuration $\mb{X}(t)$ in the notation introduced in section \ref{sec:reduced}. More specifically, we call the gel band on the left our initial configuration and set the number of particles of type 1 to the number of bands, 5, and the number of particles of type 2 to 0. 
On the right set of diagrams, a birth event keeps the number of type 1 particles intact and increments the number of type 2 particles by one, a death event changes the number of type 1 particles from five to four and keeps the number of type 2 particles at zero, and finally a shift event decreases the number of type 1 particles by one and increases the number of type 2 particles by one.} 
\label{fig:BDS}
\end{figure}
Estimating the rates based on observed changes at genotyping times in this experimental setup corresponds to inference in a discretely observed linear BDS process.
That is, we assume each element behaves independently, and that overall rates of each event are proportional to total copy number $k$. Together with the time-homogeneity assumption, waiting times until occurrence of an event are distributed exponentially with rate $k\theta$, where $\theta = \lambda + \mu + \nu$. When an event occurs, the probability that it is birth, death, or shift is given by $\lambda/\theta, \nu/\theta,$ and $\mu/\theta$ respectively. The BDS process is therefore a continuous-time Markov chain (CTMC). 

The states in our process $\widetilde{\mb{x}} \in \left\{0,1 \right\}^S := \widetilde\Omega $ can be represented as binary vectors, where $S$ is the number of possible locations transposons may occupy along the genome, $0$'s denote unoccupied sites, and $1$'s correspond to sites occupied by a transposon. Now, denote the $2^S \times 2^S$ rate matrix or infinitesimal generator corresponding to this CTMC as $\mathbf{Q} = \left\{ q_{\widetilde{\mb{x}}_1, \widetilde{ \mb{x}}_2} \right\}$, where $q_{\widetilde{\mb{x}}_1, \widetilde{ \mb{x}}_2} $ denotes the instantaneous rate of jumping to $\widetilde{\mb{x}}_2$ beginning from $\widetilde{\mb{x}}_1$ with $\widetilde{\mb{x}}_1, \widetilde{\mb{x}}_2 \in \widetilde\Omega$. 
To write down the entries of $\mathbf{Q}$, first define $C^+(\widetilde{\mb{x}})$ as the set of all configurations with one additional site occupied relative to $\widetilde{\mb{x}}$. Thus, $C^+(\widetilde{\mb{x}})$ contains states corresponding to one birth event beginning with $\widetilde{\mb{x}}$. Similarly $C^\rightarrow(\widetilde{\mb{x}})$ contains states where one additional site is occupied and one originally occupied site is no longer occupied, and $C^-(\widetilde{\mb{x}})$ contains states where one originally occupied site in $\widetilde{\mb{x}}$ is no longer occupied. 
Then $| C^+(\widetilde{\mb{x}}_1) | = S - k$, $| C^-(\widetilde{\mb{x}}_1) | = k$, and $| C^\rightarrow(\widetilde{\mb{x}}_1) | = | C^+(\widetilde{\mb{x}}_1) | \times | C^-(\widetilde{\mb{x}}_1) |$, and finally the entries of the generator $\mathbf{Q}$ are given by
\begin{equation}\label{eq:Q}
q_{\widetilde{\mb{x}}_1, \widetilde{ \mb{x}}_2} = \frac{\lambda}{| C^+(\widetilde{\mb{x}}_1) |} \indfun{\widetilde{\mb{x}}_2 \in C^+(\widetilde{\mb{x}}_1)} + \frac{\nu}{| C^\rightarrow(\widetilde{\mb{x}}_1) |} \indfun{\widetilde{\mb{x}}_2 \in C^\rightarrow(\widetilde{\mb{x}}_1)}  
+ \frac{\mu}{| C^-(\widetilde{\mb{x}}_1) |} \indfun{\widetilde{\mb{x}}_2 \in C^-(\widetilde{\mb{x}}_1)}. 
\end{equation}

\subsection{BDS process with covariates}

We are interested in inference when the data consist of $m$ independent processes $\left\{ \widetilde{\mathbf{X}}^{p}(t) \right\}$, $p = 1, \ldots, m$, each discretely observed at times $0 = t_{p,0} < t_{p,1} < \ldots < t_{p,n(p)}$. We assume each $\{\widetilde{\mathbf{X}}^{p}(t) \}$ process evolves according to a linear BDS model with per-particle instantaneous birth rate $\lambda_p \geq 0$, shift rate $\nu_p \geq 0$, and death rate $\mu_p \geq 0$. 
The data, observations from each process, are points in the previously defined state space, 
with $\widetilde{\mathbf{X}}^{p}(t) \in \widetilde\Omega$ for any fixed $p$ and $t$. 
For example, in transposon evolution, each patient $p$ is genotyped at $n(p)+1$ observation times, and at each given time, the $1$'s present in the data vector correspond to locations in the gel currently occupied by transposons. 
The observed data corresponding to a given process $\{\widetilde{\mathbf{X}}^{p}(t) \}$ can thus be collected in a $S \times [n(p)+1]$ matrix with columns corresponding to observation times, and the full observed dataset can be collected into a $S \times \sum_{p=1}^m [n(p)+1]$ matrix: 
$\mathbf{Y} = (\widetilde{\mathbf{X}}^{1}(t_{1,0}), \ldots, \widetilde{\mathbf{X}}^{1} ( t_{1,n(1)}) , \ldots, \widetilde{\mathbf{X}}^{m}( t_{m,0}), \ldots, \widetilde{\mathbf{X}}^{m} (t_{m,n(m)})).$  

The rates of each process are determined by a vector of $c$ covariates $\mb{z}_p = (z_{p,1}, z_{p,2}, \ldots, z_{p, c} ) \in \mathbb{R}^c$ through a log-linear model: 
\begin{equation}\label{eq:loglin}
  \log(\lambda_p) = \boldsymbol\beta^\lambda \cdot \mathbf{z_p},  \quad \log(\nu_p) = \boldsymbol\beta^\nu \cdot \mathbf{z_p}, \quad\log(\mu_p) = \boldsymbol\beta^\mu \cdot \mathbf{z_p},
  \end{equation}
where  $\boldsymbol\beta := (\boldsymbol\beta^\lambda,  \boldsymbol\beta^\nu,  \boldsymbol\beta^\mu)$ are the regression coefficients and  $\cdot$  represents a vector product. For instance, in an epidemiological study, these covariates may contain patient-specific disease process and demographic information.

The observed data log-likelihood is obtained by summing over transition terms of observations for each process, and summing over all processes:
\begin{equation}\label{eq:oll}
\widetilde{\ell_o}(\mb{Y} ; \boldsymbol\beta) = \sum_{p=1}^m \sum_{j = 0}^{n(p)-1} \log \widetilde{p}_{\widetilde{\mathbf{X}}^{p} ( t_{p,j}) , \widetilde{\mathbf{X}}^{p} ( t_{p,j+1})} ( t_{p, j+1} - t_{p,j} ; \lambda_p, \nu_p, \mu_p),
 \end{equation}
where $\widetilde{p}_{\widetilde{\mathbf{ x}}_1,\widetilde{ \mathbf {x}}_2}(t ; \lambda, \nu, \mu) = \text{Pr}_{\lambda, \nu, \mu} (\widetilde{\mathbf{X}}(t)= \widetilde{\mathbf{x}}_2 \mid \widetilde{\mathbf{X}}(0) = \widetilde{\mathbf{ x}}_1 )$ denotes a transition probability of the BDS process.
We are interested in computing the maximum likelihood estimates (MLEs) of parameters $\boldsymbol\beta$ of the BDS process. Notice that if the transition probabilities were available for given $\lambda, \nu, \mu, $ and $t$ values, one could maximize the likelihood in \eqref{eq:oll} using standard off-the-shelf optimization procedures. However, due to the large state space of all possible configurations of occupied sites, analysis of these transition probabilities is intractable.
To approximate the BDS model likelihood above, we introduce a two-type branching process such that computationally tractable transition probabilities of this process are numerically close to the transition probabilities of the BDS model over any observation interval. 
The following sections detail the correspondence between the BDS model and the two-type branching process, and develops methodology for inference  in the branching process framework.


\subsection{Reducing the state space}\label{sec:reduced}
The size of the original state space $ |\widetilde\Omega | = 2^S$ quickly becomes unmanageable as $S$ grows
so that analysis using the rate matrix defined in \eqref{eq:Q} becomes unwieldy for all but small values of $S$.
Previous work by \cite{doss2013} addresses this issue by collapsing the state space to one dimension, distilling the data to copy number counts at each observation time. In this simplified setting, they develop tools for inference in a discretely observed birth-death-immigration framework.  However, this approximate model ignores particle shifts which do not affect the total copy number, rendering the shift rate unidentifiable. Further, collapsing the state space in this way violates the Markov assumption in the BDS model. In particular, waiting times between birth and death events are exponentially distributed under the model in \cite{doss2013}, but under the BDS model with shift events, the waiting time between a birth and death no longer follows an exponential distribution. 


Instead of ignoring shifts, we propose a reduction of the state space into a two-dimensional representation $\Omega \in \mathbb{N} \times \mathbb{N}$.
Elements of this reduced space are pairs $\mb{X}(t)= (x_{old}, x_{new}) \in \Omega$ tracking the number of originally occupied and newly occupied sites at the end of each observation interval. As an example, assume six particles are present initially at time $t_0$,  and a shift and a birth occur before the first observation $t_1$, and a death occurs before a second observation at $t_2$. When considering the first observation interval $[t_0,t_1)$, we have $ \left\{ \mb{X}(t_0) = (6,0), \mb{X}(t_1)=  (5,2) \right\}$. When computing the next transition probability over $[t_1, t_2)$, we now have $\left\{ \mb{X}(t_1) = (7,0), \mb{X}(t_2) = (6,0) \right\}$, since all seven of the particles at $t_1$, now the left endpoint of the observation interval, now become the initial population. This seemingly inconsistent definition of the state at $\mb{X}(t_1)$ is not a problem: we will see that all necessary computations occur separately on disjoint intervals, so that our reduced representation of the original process needs only to be defined consistently for any given pair of consecutive observations.

Formally, this state space transformation is a mapping $\psi: \widetilde\Omega \times \widetilde\Omega \rightarrow \Omega \times \Omega$ on consecutive pairs of observations in $\widetilde\Omega$ to the reduced state space that can be computed
$\psi: \left\{ \widetilde{\mathbf{X}}(t_1), \widetilde{\mathbf{X}}(t_2) \right\} \mapsto \left\{ (a, 0) , (b,c) \right\}  = \{ \mathbf{X}(t_1), \mathbf{X}(t_2) \}$, 
where $a =  \sum_{j=1}^S \widetilde{X}_{j}(t_1)$ is the total number of initially occupied sites in $\widetilde{\mathbf{X}}(t_1)$, $b = \sum_{j=1}^S \indfun{\widetilde X_j(t_2) = \widetilde X_j(t_1)} \widetilde X_j(t_1) $ is the number of initially occupied sites that remain occupied, and $c = \sum_{j=1}^S \indfun{\widetilde X_j(t_2) - \widetilde X_j(t_1) = 1}$ is the number of newly occupied sites in $\widetilde{\mathbf{X}}(t_2)$ not present in $\widetilde{\mathbf{X}}(t_1)$. 


Note that while $\psi$ significantly reduces the size of the state space, the mapping discards information about specific particle locations, which is uninformative to inferring birth, death, and shift rates due to symmetry induced by particle independence. The number of changes in locations between observations --- the data relevant to our estimation task --- is preserved in the image of $\psi$. 

\subsection{Modeling via two-type branching process}
Working now in the space $\Omega$, we can treat $x_{old}$ and $x_{new}$ as particle types in a two-type branching process.
Let $a_{j}(k,l)$ be the rate of producing $k$ type 1 particles and $l$ type 2 particles, beginning with one type $j$ particle, $j = 1,2$. Then the nonzero rates defining the two-type branching process corresponding to the birth-death-shift model are given by
\begin{align}
a_1(1,1) &= \lambda, && a_1(0,1) = \nu, && a_1(0,0) = \mu, && a_1(1,0) = -(\lambda + \nu + \mu),  \label{eq:rates} \\
a_2(0,2) &= \lambda, && a_2(0,1) = -(\lambda + \mu), && a_2(0,0) = \mu. \nonumber
\end{align}
This characterization enables us to apply a generating function approach to calculate transition probabilities of the process. Defining $X_j(t)$, the number of particles of type $j$ at time $t$, we consider the generating function 
\begin{equation}\label{eq:probgen}
\phi_{jk}(t, s_1, s_2) = \text{E} \left(s_1^{X_1(t)} s_2^{X_2(t)} \mid X_1(0) = j, X_2(0) = k \right) = \sum_{l=0}^\infty \sum_{m=0}^\infty p_{(j,k), (l,m)} (t) s_1^l s_2^m. 
\end{equation}
Using the Kolmogorov backward equations, we derive equations and a closed form solution for $\phi_{jk}$ (see Appendix A).  Although an analytical expression is available, it involves special functions that are in practice are often unstable. Instead, we simplify the backward equations so that evaluating $\phi_{jk}$ only requires solving a single linear ordinary differential equation, which is easily accomplished using standard Runge-Kutta methods \citep{butcher1987}.

With $\phi_{jk}$ available, we see from \eqref{eq:probgen} that the transition probabilities $p_{(j,k), (l,m)} (t)$ can then be obtained by differentiating and setting $s_1, s_2 = 0$, but without an analytical expression for these derivatives, repeated numerical differentiation is inefficient and numerically unstable. Instead, we follow \citet{lange1982} and \citet{doss2013} and map our domain $[0,1] \times [0,1]$ to the boundary of a unit circle in the complex plane by setting $s_1 = e^{2 \pi i w_1}, s_2 = e^{2 \pi i w_2}$. Under this change of variables, the generating function becomes a Fourier series
\[\phi_{jk}\left(t, e^{2 \pi i w_1}, e^{2 \pi i w_2} \right)= \sum_{l,m = 0}^\infty p_{(j,k), (l,m)}(t) e^{2 \pi i l w_1} e^{2 \pi i m w_2} . \]
Applying a Riemann sum approximation to the integral corresponding to coefficients given by the Fourier inversion formula, we can compute the transition probabilities using
\begin{equation}\label{eq:FFT}
\begin{split}
 p_{(j,k),(l,m)}(t) &= \int_0^1 \int_0^1 \phi_{jk}(t, e^{2 \pi i w_1}, e^{2 \pi i w_2}) e^{-2 \pi i l w_1} e^{-2 \pi i m w_2} dw_1 dw_2  \\
 & \approx  \frac{1}{N^2} \sum_{u = 0}^{N-1} \sum_{v= 0}^{N-1} \phi_{jk}(t, e^{2 \pi i u/N}, e^{2 \pi i v/N}) e^{-2 \pi i l u/N} e^{-2 \pi i m v/N} .
 \end{split}
 \end{equation}
Choice of a larger $N$ leads to a finer and thus more accurate Riemann sum approximation of the integral, and also allows us to compute transition probabilities to and from a larger total particle population of either type.
The Fast Fourier transform (FFT) enables efficient computation of these coefficients \citep{henrici1979}, and in our application and simulation studies, we find that a grid size as small as $N=16$ yields accurate results. With transition probabilities available, we may closely approximate $\widetilde{\ell_o}(\mb{Y} ; \boldsymbol\beta) $ by the branching process likelihood
\begin{equation}\label{eq:approxobservedlike}
\ell_o(\mb{Y} ; \boldsymbol\beta) = \sum_{p=1}^m \sum_{j = 0}^{n(p)-1} \log p_{\mathbf{X}^{p} ( t_{p,j}), \mathbf{X}^{p} ( t_{p,j+1})} ( t_{p, j+1} - t_{p,j} ; \lambda_p, \nu_p, \mu_p),  \end{equation}
so that maximizing the observed likelihood in \eqref{eq:oll} gives approximately the same 
parameter estimates as maximizing \eqref{eq:approxobservedlike}. 

\subsection{EM algorithm for the BDS process}
With transition probabilities of the process available, it is already possible to produce MLEs of the covariate effects associated with birth, death, and shift rates by numerical maximization of the observed likelihood. However, an EM algorithm approach often outperforms off-the-shelf optimization procedures in missing data problems, offering a significantly faster and more robust solution.
Let $\ell_c ( \mb{X}, \boldsymbol\beta)$ denote the complete data log-likelihood, $\mb{X}$ the complete data, and $\mb{Y}$  the available observations. The EM algorithm begins with an initial parameter estimate $\boldsymbol\beta_0$, and then at each $j^\text{th}$ iteration, updates the estimate by setting
\begin{equation}\label{eq:EM}
\boldsymbol\beta_j = \argmax_{\boldsymbol\beta} \text{E}_{\boldsymbol\beta_{j-1}} \left[ \ell_c (\mb{X}, \boldsymbol\beta) \mid \mb{Y} \right].
\end{equation}
Each iteration involves a computation of the expectation term called the \textit{E-step}, followed by a maximization of the expectation called the \textit{M-step}. 

 \subsubsection{E-step}
The fully observed BDS process is a continuous-time Markov chain, so its complete-data log-likelihood can be written as 
\begin{equation}\label{eq:completell}
 \ell_c(\mb{X}; \boldsymbol\beta) = \sum_{p=1}^m \left[ b_p \log \lambda_p + f_p \log \nu_p + d_p \log \mu_p - (\lambda_p + \mu_p + \nu_p) \sum_{k=0}^{\infty} k \tau_p(k) + \sum_{k=0}^{\infty} \log \tau_p(k) \right], 
 \end{equation}
where $\tau_p(k)$ is the total time process $\mb{X}^p(t)$ spends with total copy number $X^p_1(t) + X^p_2(t) = k$, $b_p$ is the total number of births, $f_p$ the number of shifts, and $d_p$ the number of deaths for each patient $p = 1 , \ldots, m$ --- these quantities are the complete data sufficient statistics \citep{guttorp1995}. Notice the final term in \eqref{eq:completell} is constant with respect to the parameters. We see that in order to obtain the expected complete-data log-likelihood, we need  to calculate only expected births --- $\text{E}_{\boldsymbol\beta}\left[ b_p \mid \mb{Y} \right]$, shifts --- $\text{E}_{\boldsymbol\beta}[f_p \mid \mb{Y} ]$, deaths --- $\text{E}_{\boldsymbol\beta}[d_p \mid \mb{Y}]$, and particle time --- $\text{E}_{\boldsymbol\beta} \left[R_p \mid \mb{Y} \right]$, where the last quantity is defined as
\[\text{E}_{\boldsymbol\beta} \left[R_p \mid \mb{Y} \right] :=  \text{E}_{\boldsymbol\beta} \left[ \int _{t_{p,0}}^{t_{p,n(p)}} X_1(s)+X_2(s) ds \mid \mb{Y} \right] =  \text{E}_{\boldsymbol\beta}\left[ \sum_{k=0}^{\infty} k \tau_p(k) \mid \mb{Y} \right]. \]
By independence of the $p$ processes and linearity of expectations, each expectation breaks into sums of expectations over the observation intervals. Further, by homogeneity, it suffices to be able to calculate the quantities
 \[ e^+ _{jk,lm} (t) = \text{E} \left[ b_{p,t} \mid \mb{X}^p(0) = (j,k), \mb{X}^p(t) = (l,m) \right], \]
\[ e^{\rightarrow}_{jk,lm} (t) = \text{E} \left[ f_{p,t} \mid \mb{X}^p(0) = (j,k), \mb{X}^p(t) = (l,m) \right], \]
\[ e^- _{jk,lm} (t) = \text{E} \left[ d_{p,t} \mid \mb{X}^p(0) = (j,k), \mb{X}^p(t) = (l,m) \right], \]
\[ e^*_{jk,lm} (t) = \text{E} \left[ R_{p,t} \mid \mb{X}^p(0) = (j,k), \mb{X}^p(t) = (l,m) \right], \]
for all non-negative integers $j,k,l,m$. Dependence of these quantities on rates $\lambda_p, \nu_p, \mu_p$ is  suppressed in the notation here for simplicity.
As noticed by  \citet{minin2008} and \citet{doss2013}, it is easier to work via the restricted moments
\begin{align*}
m^+_{jk,lm}(t) &= \text{E} \left[ b_{p,t} 1_{ \{\mb{X}^p(t) = lm \} } \mid \mb{X}^p(0) = (j,k) \right] = \sum_{n=0}^\infty n q^+_{jk,lm} (n,t),\\
m^\rightarrow_{jk,lm}(t) &= \text{E} \left[ f_{p,t} 1_{ \{\mb{X}^p(t) = lm \} } \mid \mb{X}^p(0) = (j,k) \right] = \sum_{n=0}^\infty n q^\rightarrow_{jk,lm} (n,t),\\
m^-_{jk,lm}(t) &= \text{E} \left[ d_{p,t} 1_{ \{\mb{X}^p(t) = lm \} } \mid \mb{X}^p(0) = (j,k) \right] = \sum_{n=0}^\infty n q^-_{jk,lm} (n,t),\\
m^*_{jk,lm}(t) &= \text{E} \left[ R_{p,t} 1_{ \{\mb{X}^p(t) = lm \} } \mid \mb{X}^p(0) = (j,k) \right] = \int_{x=0}^\infty x \mathrm{d}  q^*_{jk,lm} (x,t),
\end{align*}
where 
\begin{align*}
q^*_{jk,lm}(x,t) &= \text{Pr}[R_{p,t} \leq x, \mb{X}^p(t) = (l,m) \mid \mb{X}^p(0)= (j,k)],\\
q^+_{jk,lm}(n,t) &= \text{Pr}[ b_{p,t} = n, \mb{X}^p(t) = (l,m) \mid \mb{X}^p(0) = (j,k)],
\end{align*}
and $q^\rightarrow, q^-$ are defined analogously. The conditional expectations can then be recovered after dividing by transition probabilities, i.e.
\[ e^+_{jk,lm}(t) = m^+_{jk,lm}(t)/ p_{jk,lm}(t). \]
These restricted moments can be computed with a similar approach used to obtain transition probabilities. We begin by defining the pseudo-generating functions: for expected births, let
\[ g^+_{jk,lm}(r,t) = \sum_{n=0}^\infty q^+_{jk,lm}(n,t) r^n. \]
Ignoring notational dependence on individual patients for simplicity, we define the joint generating function
\begin{align*}
&H^+_{jk}(r,s_1, s_2, t) = \text{E} \left[ r^{b_t}  s_1^{X_1(t)} s_2^{X_2(t)} \mid \mb{X}(0) = (j,k) \right] \\
	&= \sum_l \sum_m \sum_n \text{P}r\left[b_t = n , \mb{X}(t) = (k,l) \mid \mb{X}(0) = (j,k) \right] r^n s_1^l s_2^m 
	= \sum_l \sum_m g^+_{jk,lm} (r,t) s_1^l s_2^m  .
\end{align*}
Pseudo-generating functions for shifts and deaths are defined analogously, and the pseudo-generating function for particle time is defined as
\begin{align*}
H^*_{jk}(r,s_1, s_2, t) &= \sum_l \sum_m \int_{x=0}^\infty  e^{-rx} \mathrm{d} q^*_{jk,lm}(x,t)  s_1^l s_2^m := \sum_l \sum_m V_{jk,lm} (r,t) s_1^l s_2^m ,
\end{align*}
where $V_{jk,lm}(r;t) = \int_0^\infty e^{-rx} \mathrm{d} q^*_{jk,lm}(x;t)$ is the Laplace-Stieltjes transform of $q^*_{jk,lm}(x;t)$.
In each case we can define series whose coefficients are our quantities of interest by partial differentiation:
\begin{equation}\label{eq:partialBirths}
G_{jk}^+(s_1, s_2, t) = \frac{d}{dr} H_{jk}^+ (r, s_1, s_2, t) \bigg|_{r=1} = \sum_l \sum_m \left[ \sum_n n q^+_{jk,lm}(n,t) \right] s_1^l s_2^m = \sum_l \sum_m m^+_{jk,lm}(t) s_1^l s_2^m .
\end{equation}
$G_{jk}^\rightarrow$ and $G_{jk}^-$ are defined analogously, and the expression for particle time is instead differentiated at $r=0$:
\begin{equation}\label{eq:partialParticleTime}
G_{jk}^*(s_1, s_2, t) = \frac{d}{dr} H_{jk}^* (r, s_1, s_2, t) \bigg|_{r=0} = \sum_l \sum_m \left[ \int_{x=0}^\infty x \mathrm{d}  q^*_{jk,lm} (x,t) \right] s_1^l s_2^m = \sum_l \sum_m m^*_{jk,lm}(t) s_1^l s_2^m .
\end{equation}
We see that given expressions for $H^+_{jk}$, $H^\rightarrow_{jk}$, $H^-_{jk}$, and $H^*_{jk}$, the coefficients corresponding to moments $m^+_{jk,lm}$, $m^\rightarrow_{jk,lm}$, $m^-_{jk,lm}$,  $m^*_{jk,lm}$ can then be numerically computed using FFT analogously to \eqref{eq:FFT} by replacing $\phi_{jk}$ with the corresponding $G_{jk}$ functions. For notational simplicity, we use $G_{jk}$ when referring collectively to $G^+_{jk}$, $G^\rightarrow_{jk}$, $G^-_{jk}$, and $G^*_{jk}$, and similarly define $H_{jk}$.
\par
Having reduced our task to computing $H_{jk}$, we define $H_1 := H_{10}(r,s_1, s_2, t)$ and  $H_2 := H_{01}(r,s_1, s_2, t)$.
By independence of particles in the branching process, we have $H_{jk} = H_1^j H_2^k$. In all four cases, $H_2$ is analytically available, and we derive an ordinary differential equation for $H_1$, summarized in the theorem below. 
  We present the result for a branching process with rates corresponding to the birth-death-shift model, but such systems of equations are available for an arbitrary time-homogeneous multi-type branching process.
\begin{theorem}\label{thm:main}
Let $\{ X_t \}$ be a two-type branching defined by the rates in \eqref{eq:rates}. Denote particle time and the number of births, shifts, and deaths over the interval $[0,t)$ by $R_t, b_t, f_p$, and $d_t$ respectively. Define the generating functions corresponding to births as 
\begin{align*}
H^+_1(r,s_1, s_2, t) &= \text{E} \left[ r^{b_t}  s_1^{X_1(t)} s_2^{X_2(t)} \mid \mb{X}(0) = (1,0) \right] \text{ and}\\
H^+_2(r,s_1, s_2, t) &= \text{E} \left[ r^{b_t}  s_1^{X_1(t)} s_2^{X_2(t)} \mid \mb{X}(0) = (0,1) \right].
\end{align*}
Then
\[
H_2^+ = y_b + \left[  \frac{-\lambda r}{2 \lambda r y_b - \lambda - \mu} + \left( \frac{1}{s_2 - y_b} + \frac{\lambda r}{2 \lambda r y_b - \lambda - \mu} \right)  e^{-(2y_b \lambda r - \lambda - \mu)t} \right] ^{-1},
\]
where $y_b = (\lambda + \mu + \sqrt{ \lambda^2 + 2 \lambda \mu + \mu^2 - 4 \lambda \mu r})/(2 \lambda r)$, and $H_1^+$ satisfies the following differential equation:
\begin{equation}\label{eq:birthODE}
\frac{d}{dt} H_1^+(t,s_1,s_2,r) = \lambda r H_1^+ H_2^+ + \nu H_2^+ + \mu - (\lambda + \mu + \nu) H_1^+, 
\end{equation}
subject to initial condition $H_1(r,s_1,s_2,0) = s_1$.

The analogous generating functions for shifts, deaths, and particle time satisfy the following equations:
\begin{align*}
H_2^- (t,s_1,s_2,r) &= y_d + \left[  \frac{-\lambda }{2 \lambda y_d - \lambda - \mu} + \left( \frac{1}{s_2 - y_d} + \frac{\lambda }{2 \lambda y_d - \lambda - \mu} \right)  e^{-(2y_d \lambda  - \lambda - \mu)t} \right] ^{-1},\\
H_2^\rightarrow (t,s_1,s_2,r) &= 1 + \left[ \frac{\lambda}{\mu - \lambda} + (\frac{1}{s_2-1} + \frac{\lambda}{\lambda - \mu})e^{(\mu - \lambda)t} \right] ^{-1}, \\
H_2^*(t,s_1,s_2,r) &= y_* + \left[  \frac{-\lambda }{2 \lambda y_* - \lambda - \mu - r} + \left( \frac{1}{s_2 - y_*} + \frac{\lambda }{2 \lambda y_* - \lambda - \mu - r} \right)  e^{-(2y_* \lambda  - \lambda - \mu - r)t} \right] ^{-1} ,
\end{align*}
\begin{align*}
\frac{d}{dt} H_1^- (t,s_1,s_2,r) &=  \lambda H_1^- H_2^- + \nu H_2^- + \mu r - (\lambda + \mu + \nu) H_1^-, \\
\frac{d}{dt} H_1^\rightarrow (t,s_1,s_2,r) &=  \lambda H_1^\rightarrow H_2^\rightarrow + \nu r H_2^\rightarrow + \mu  - (\lambda + \mu + \nu) H_1^\rightarrow,  \\
\frac{d}{dt} H_1^* (t,s_1,s_2,r) &=  \lambda H_1^* H_2^* + \nu H_2^* + \mu - (\lambda + \mu + \nu + r) H_1^*,
\end{align*}
where $y_d = (\lambda + \mu + \sqrt{ \lambda^2 + 2 \lambda \mu + \mu^2 - 4 \lambda \mu r})/(2 \lambda)$, 
$y_* = (\lambda + \mu + r +  \sqrt{( \lambda + \mu + r)^2 - 4 \lambda \mu})/(2 \lambda)$, and 
$H_1^-(r,s_1,s_2,0) = H_1^\rightarrow(r,s_1,s_2,0) = H_1^*(r,s_1,s_2,0) = s_1$.

%
%

\end{theorem}

\begin{proof} The derivations are included in the Appendix B, using the birth equations \eqref{eq:birthODE} as a detailed example. The other systems follow analogous derivations.  \end{proof}


This theorem shows that for each of the necessary sufficient statistics, computations for $H_{jk}$ are essentially reduced to solving a single ordinary differential equation. As discussed for transition probabilities, this is easily accomplished using Runge-Kutta methods, which in practice offer more numerical stability than working with solutions obtained analytically by integrating the ODE. Because we can evaluate $H_{jk}$, we can also easily differentiate $H_{jk}$ numerically, yielding access to numerical solutions to $G_{jk}$.

To summarize, with $H^+_{jk}, H^-_{jk}, H^\rightarrow_{jk} , H^*_{jk}$ now available, we may obtain the restricted moments by computing the coefficients in the power series $G^+_{jk}, G^-_{jk}, G^\rightarrow_{jk} , G^*_{jk}$. These coefficients are recovered using a Riemann approximation to the Fourier inversion formula analogous to formula \eqref{eq:FFT}. For instance,
\[ \begin{split}
 m^+_{(jk),(lm)}(t) &= \int_0^1 \int_0^1 G^+_{jk}\left(t, e^{2 \pi i t_1}, e^{2 \pi i t_2}\right) e^{-2 \pi i l t_1} e^{-2 \pi i m t_2} dt_1 dt_2  \\
 & \approx  \frac{1}{N^2} \sum_{u = 0}^{N-1} \sum_{v= 0}^{N-1} G^+_{jk}\left(t, e^{2 \pi i u/N}, e^{2 \pi i v/N}\right) e^{-2 \pi i l u/N} e^{-2 \pi i m v/N} .
 \end{split} \]
 
We are thus able to compute all necessary quantities appearing in the expected complete-data log-likelihood $\text{E}_{\widetilde{\boldsymbol\beta}} \left[ \ell_c (\mb{X}, \boldsymbol\beta) \mid \mb{Y} \right] $. Recall that sufficient statistics for each patient $b_p, f_p, d_p,$ and $R_p$ break up over intervals: i.e.\ the total number of births $b_p$ is equal to the sum of the number of births over each disjoint interval $[ t_{p,j-1}, t_{p,j} )$, with $j = 1, \ldots, n(p)$. Further, by the Markov property, the conditional expectation of the number births over an interval $[ t_1, t_2)$ given $\mb{Y}$ depends only on the states of the process at the endpoints of the interval: 
\begin{equation}\label{eq:markovconditional}
 \text{E} \left[ b_{p,t_2-t_1} \mid \mb{Y} \right] = \text{E} \left[ b_{p,t_2-t_1} | \mb{X}^p(t_1), \mb{X}^p(t_2) \right] = e^+_{\mb{X}^p(t_1), \mb{X}^p(t_2)}(t_2-t_1) = \frac{ m^+_{\mb{X}^p(t_1), \mb{X}^p(t_2)}(t_2-t_1)} { p_{\mb{X}^p(t_1), \mb{X}^p(t_2)}(t_2-t_1)}   , 
 \end{equation}
and the same is true for the other sufficient statistics. Therefore, for each process $p$,
\begin{align}
 \text{E}_{\widetilde{\boldsymbol\beta}} \left[ b_p \mid \mb{Y} \right] &= \sum_{i=1}^{n(p)}  e^+_{\mb{X}^p(t_{p,i-1}), \mb{X}^p(t_{p,i})}(t_{p,i-1}-t_{p,i}; \widetilde{\lambda}_p, \widetilde{\nu}_p, \widetilde{\mu}_p ), \nonumber \\
 \text{E}_{\widetilde{\boldsymbol\beta}} \left[ f_p \mid \mb{Y} \right] &= \sum_{i=1}^{n(p)}  e^\rightarrow_{\mb{X}^p(t_{p,i-1}), \mb{X}^p(t_{p,i})}(t_{p,i-1}-t_{p,i}; \widetilde{\lambda}_p, \widetilde{\nu}_p, \widetilde{\mu}_p ), \text{ and} \label{eq:suffstat}\\
 \text{E}_{\widetilde{\boldsymbol\beta}} \left[ d_p \mid \mb{Y} \right] &= \sum_{i=1}^{n(p)}  e^-_{\mb{X}^p(t_{p,i-1}), \mb{X}^p(t_{p,i})}(t_{p,i-1}-t_{p,i}; \widetilde{\lambda}_p, \widetilde{\nu}_p, \widetilde{\mu}_p ), \nonumber
 \end{align}
with $  \log\left(\widetilde{\lambda}_p\right) = \widetilde{\boldsymbol\beta}^\lambda \cdot \mathbf{z_p},  \log\left(\widetilde{\nu}_p\right) = \widetilde{\boldsymbol\beta}^\nu \cdot \mathbf{z_p}, \log\left(\widetilde{\mu}_p\right) = \widetilde{\boldsymbol\beta}^\mu \cdot \mathbf{z_p} $ similarly to equation \eqref{eq:loglin}.
Finally, combining \eqref{eq:suffstat}, \eqref{eq:markovconditional}, and \eqref{eq:completell}, the expected complete-data log likelihood   up to a constant is equal to
\begin{align}
\text{E}_{\widetilde{\boldsymbol\beta}} \left[ \ell_c (\mb{X}, \boldsymbol\beta) \mid \mb{Y} \right]  \propto \sum_{p=1}^m \bigg\{ \sum_{j=1}^{n(p)}  & \bigg[ \frac{ m^+_{\mb{X}^p(t_{p,j-1}), \mb{X}^p(t_{p,j})}(t_{p,j}-t_{p,j-1}; \widetilde{\lambda}_p, \widetilde{\nu}_p, \widetilde{\mu}_p)} { p_{\mb{X}^p(t_{p,j-1}), \mb{X}^p(t_{p,j})}(t_{p,j}-t_{p,j-1}; \widetilde{\lambda}_p, \widetilde{\nu}_p, \widetilde{\mu}_p)} \log \lambda_p  \nonumber \\
 & + \frac{ m^\rightarrow_{\mb{X}^p(t_{p,j-1}), \mb{X}^p(t_{p,j})}(t_{p,j}-t_{p,j-1}; \widetilde{\lambda}_p, \widetilde{\nu}_p, \widetilde{\mu}_p)} { p_{\mb{X}^p(t_{p,j-1}), \mb{X}^p(t_{p,j})}(t_{p,j}-t_{p,j-1}; \widetilde{\lambda}_p, \widetilde{\nu}_p, \widetilde{\mu}_p)} \log \nu_p  \nonumber \\
& + \frac{ m^-_{\mb{X}^p(t_{p,j-1}), \mb{X}^p(t_{p,j})}(t_{p,j}-t_{p,j-1}; \widetilde{\lambda}_p, \widetilde{\nu}_p, \widetilde{\mu}_p)} { p_{\mb{X}^p(t_{p,j-1}), \mb{X}^p(t_{p,j})}(t_{p,j}-t_{p,j-1}; \widetilde{\lambda}_p, \widetilde{\nu}_p, \widetilde{\mu}_p)} \log \mu_p  \nonumber \\
& -  \frac{ m^*_{\mb{X}^p(t_{p,j-1}), \mb{X}^p(t_{p,j})}(t_{p,j}-t_{p,j-1}; \widetilde{\lambda}_p, \widetilde{\nu}_p, \widetilde{\mu}_p)} { p_{\mb{X}^p(t_{p,j-1}), \mb{X}^p(t_{p,j})}(t_{p,j}-t_{p,j-1}; \widetilde{\lambda}_p, \widetilde{\nu}_p, \widetilde{\mu}_p)} (\lambda_p + \mu_p + \nu_p) \bigg] \bigg\}. \label{eq:recapll}  
 \end{align}

\subsubsection{M-step} To complete an M-step, we use an efficient Newton-Raphson algorithm to maximize the expectation 
$g(\boldsymbol\beta) = E_{\widetilde{\boldsymbol\beta}} \left[ \ell_c (\mb{X}, \boldsymbol\beta)\mid \mb{Y}  \right] .$
Each Newton-Raphson step recursively updates parameters using the following equation:
\begin{equation}\label{eq:NR}
\boldsymbol\beta_{new} = \boldsymbol\beta_{cur} - \left[ \mathbf{H} g(\boldsymbol\beta_{cur}) \right] ^{-1} \nabla g(\boldsymbol\beta_{cur}), \end{equation}
where $\nabla g$ denotes the gradient vector and $\mathbf{H} g$ denotes the Hessian matrix of $g(\boldsymbol\beta)$. 
Fortunately, compact analytical forms for these quantities are available. First, we collect complete data sufficient statistics across processes into the following vectors:
\begin{align*}
\mb{U}^T &= \left( E_{\widetilde{\boldsymbol\beta}} \left[ b_{1, t_{1, n(1)}} | \mb{Y} \right], \ldots,  E_{\widetilde{\boldsymbol\beta}} \left[ b_{m, t_{m, n(m)}} | \mb{Y} \right] \right), \quad
\mb{V}^T = \left( E_{\widetilde{\boldsymbol\beta}} \left[ f_{1, t_{1, n(1)}} | \mb{Y} \right], \ldots,  E_{\widetilde{\boldsymbol\beta}} \left[ f_{m, t_{m, n(m)}} | \mb{Y} \right] \right),\\
\mb{D}^T &= \left( E_{\widetilde{\boldsymbol\beta}} \left[ d_{1, t_{1, n(1)}} | \mb{Y} \right], \ldots,  E_{\widetilde{\boldsymbol\beta}} \left[ d_{m, t_{m, n(m)}} | \mb{Y} \right] \right), \quad
\mb{P}^T = \left( E_{\widetilde{\boldsymbol\beta}} \left[ R_{1, t_{1, n(1)}} | \mb{Y} \right], \ldots,  E_{\widetilde{\boldsymbol\beta}} \left[ R_{m, t_{m, n(m)}} | \mb{Y} \right] \right).
\end{align*}
If we aggregate covariate vectors for each process in a $c \times p$ matrix $\mb{Z} = (\mb{z_1}, \ldots, \mb{z_m})$ and process-specific rates into vectors $\boldsymbol\lambda = (\lambda_1, \ldots, \lambda_m), \boldsymbol\nu = (\nu_1, \ldots, \nu_m), \boldsymbol\mu = (\mu_1, \ldots, \mu_m)$, then the gradient and Hessian can be expressed as
\begin{align}\label{eq:grad}
\nabla g(\boldsymbol\beta) &= ( -\mb{Z}^T \left[ \text{diag}(\bf{P}) \boldsymbol\lambda + \bf{U} \right], -\mb{Z}^T \left[ \text{diag}(\bf{P}) \boldsymbol\nu + \bf{V} \right], -\mb{Z}^T \left[ \text{diag}(\bf{P}) \boldsymbol\mu + \bf{D} \right]),\\
\label{eq:hess}
\bf{H} g(\boldsymbol\beta) &= \left( \begin{array}{ccc}
-\mb{Z}^T \text{diag}(\bf{P}) \text{diag} (\boldsymbol\lambda) \mb{Z} & \mb{0} & \mb{0} \\
\mb{0} & -\mb{Z}^T \text{diag}(\bf{P}) \text{diag} (\boldsymbol\nu) \mb{Z} & \mb{0}  \\
\mb{0} & \mb{0} & -\mb{Z}^T \text{diag}(\bf{P}) \text{diag} (\boldsymbol\mu) \mb{Z} \end{array} \right).
\end{align}
The derivation of these expressions is parallel to those presented in \citep{doss2013}.
In our experience the M-step generally converges in fewer than ten Newton-Raphson steps.
Availability of closed form solutions \eqref{eq:grad} and \eqref{eq:hess} yields very
fast execution of each Newton-Raphson step, making the computational cost of the M-step negligible compared to the E-step. Note that computing the M-step using only one Newton-Raphson step rather than iterating until convergence is sufficient to guarantee the ascent property of the EM algorithm \citep{lange1995}, but we execute multiple steps because 
this strategy does not slow down our algorithm.

\subsubsection{Accelerating E-step calculations for intervals with no change}\label{sec:efficientEM}

In our birth-death-shift application, we may avoid the relatively costly E-step calculations for some intervals by approximating the probability of observing no changes with the probability that no event occurs in the underlying complete process. This approximation is not necessary in our algorithm, but can lead to gains in computational efficiency in settings such as our application where many intervals feature no observed changes.
\begin{figure}[t]
\centering

\includegraphics[width = .8\textwidth]{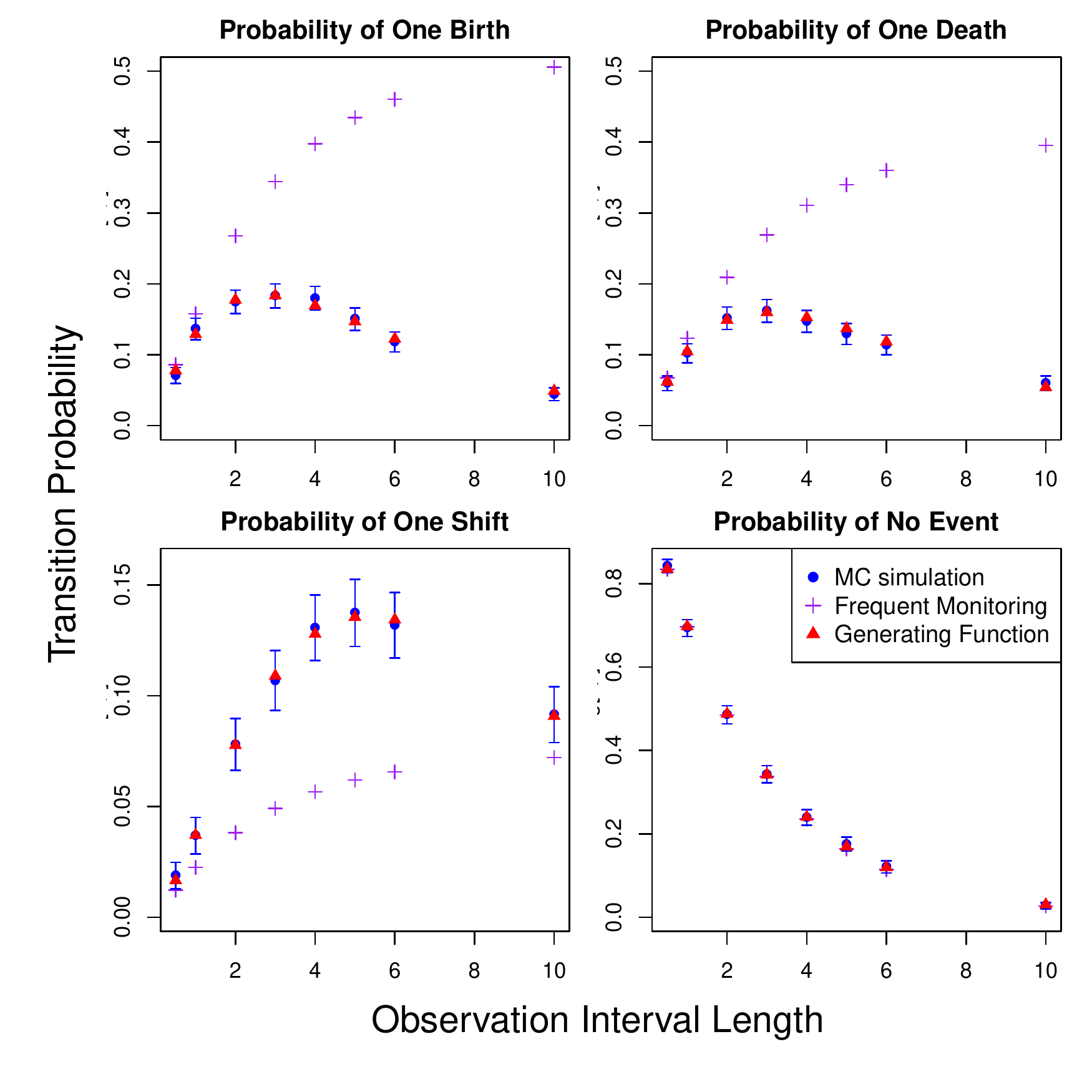}
\caption{\small Transition probability approximations. BDS 
transition probabilities are approximated with two methods --- the FM method, shown with magenta crosses, and the generating 
function  method, shown with red triangles. We depict Monte Carlo estimates of the BDS
transition probabilities with blue circles; vertical blue segments indicate their corresponding Monte Carlo confidence intervals.}
\label{fig:FMtrans}
\end{figure}

It is very unlikely that events occur in a time interval $[t_1, t_2)$ yet no change is observed so that $\mb{X}(t_1) = \mb{X}(t_2)$. For instance, if $12$ elements are present initially and a death followed by a birth occur, then we almost always observe $\mb{X}(t_1) = (12,0), \mb{X}(t_2) = (11,1)$ unless the element added by the birth occupies the \textit{exact} location that was previously occupied by the element that dies. This scenario would leave the observed state unchanged,  $\mb{X}(t_1) = \mb{X}(t_2) = (12,0)$, but has exceedingly low probability: the already small but non-negligible probability that more than one event occurs is then multiplied by $1/(S-11)$, the probability of the birth occurring in a specific location (recall $S$ is very large).
Therefore, it is numerically accurate to treat intervals with no observed changes as if no changes in the latent continuous-time process occur. In this case, the transition probability is easily calculated, given by the tail of an exponential distribution 
\begin{equation}\label{eq:tailexp}
 p_{(12,0),(12,0)}(t_2-t_1) = e^{- 12 (\lambda + \mu + \nu) (t_2-t_1)}. 
 \end{equation}
 
In addition to efficient closed-form transition probability calculation, the expected sufficient statistics necessary for the E-step are known in this setting. If no events occur, we know that $e^+_{(k,0),(k,0)}(t) = e^\rightarrow_{(k,0),(k,0)}(t) = e^-_{(k,0),(k,0)}(t) = 0$, and that the expected particle time is $e^*_{(k,0),(k,0)}(t) = k t$. This is not only faster computationally but also more numerically stable, avoiding the division of numerically calculated restricted moments by numerically calculated transition probabilities.
We verify this efficient implementation in our simulation studies, as illustrated in Figure \ref{fig:loglike}.

Equation \eqref{eq:tailexp} is the same formula that \citet{rosenberg2003} used to compute the probability of no event under their frequent monitoring (FM) method, but it is important to note that this approximation is \textit{not} accurate for the other probabilities under FM;  see Figure \ref{fig:FMtrans}. Continuing our previous example, consider the probability of transitioning from $(12,0)$ to $(11,1)$. Under the FM approximation, such a transition can only happen through a shift event with probability 
$ p_{(12,0),(11,1)}(t_2-t_1) = (\nu)/(\lambda+\nu+\mu)e^{-k (\lambda + \mu + \nu) (t_2 - t_1)}$. However, as we have discussed above, a birth followed by death also leads to observing $\mb{X}(t_2) = (11,1)$ in almost all cases. FM assigns zero probability
density to such event histories, while our method does not ignore non-negligible contribution 
of these trajectories to $ p_{(12,0),(11,1)}(t_2-t_1) $.
Additionally, probabilities $ p_{(j,k), (l,m)}(t_2-t_1)$ are set to zero under FM when for all $j,k,l,m$ where $| j-k | >1$ or $ | l-m | > 1$, preventing the use of all available data during inference. This is later illustrated in Figure \ref{fig:bias}.

\subsection{Implementation}
We implement our algorithm in the form of R package \texttt{bdsem}, available at 
\url{https://github.com/jasonxu90/bdsem}. The EM algorithm implementation relies on numerical solutions to differential equations in package \texttt{deSolve}, and accommodates panel data settings with unevenly spaced discrete observations. Our package also includes functions for MLE inference using other methods, as well as code for simulating from the BDS process. The software is accompanied by a vignette that steps through simplified versions of all simulation studies included in this paper.
 
\subsection{Comparison with frequent monitoring}
We begin with several simulation experiments assessing the validity of our algorithms.
The first simulation study checks whether transition probabilities calculated using our generating function method for the two-type branching process as described in \eqref{eq:FFT} coincide with those of the BDS model. We compare these computations to Monte Carlo estimates of these probabilities obtained from simulated trajectories from the birth-death-shift model, and also include a comparison to the FM method presented in \citep{rosenberg2003}. 

The FM model allows at most one event to occur per interval. Thus, over an observation interval $[t_i, t_{i+1})$ beginning with $k$ particles, the probabilities of a birth, death, and shift have closed forms $(\lambda/\theta) e^{-k \theta (t_{i+1} - t_i)}$, $(\mu/\theta) e^{-k \theta (t_{i+1} - t_i)}$, and $(\nu/\theta) e^{-k \theta (t_{i+1} - t_i)}$ respectively, where $\theta = \lambda + \nu + \mu$. The probability of no event occurring is given by $e^{-k \theta (t_{i+1} - t_i)}$, and all other transition probabilities are zero under the FM assumption. Because it becomes more likely that multiple events occur as the length of time between observations, $dt$, increases, we expect probabilities computed under FM to diverge substantially from the Monte Carlo estimates as we increase $dt$.
\par
We compute Monte Carlo approximations of transition probabilities from 2000 realizations of a BDS process without covariates, with rates  $\lambda = 0.0188, \mu = 0.0147, \nu = 0.00268$. These rates are equal to estimates of a transposable element birth, death, and shift rates obtained by \citet{rosenberg2003} using the FM method. We begin each simulation with an initial population size of $10$, and record the state of the process after simulating for $dt$ units of time, varying $dt$ from $0.5$ to $10$. The approximate transition probability $\hat{p}_{(10,0),(k,l)}(t)$ is then empirically computed by dividing the number of realizations ending in state $\mb{X}(t) = (k,l)$ by the total number of simulated processes.
\par
In Figure \ref{fig:FMtrans}, we see that as the length of an observation interval increases, FM approximations become inaccurate, while those obtained using our method remain within the narrow Monte Carlo confidence intervals. However, notice that the probability that no event occurs remains accurate even under the FM approximation, supporting the efficient implementation of our EM algorithm described in Section \ref{sec:efficientEM}.
Figure C-1 in the Appendix C demonstrates that our method also reliably calculates other transition probabilities that are set to 0 by the FM method, and these computations remain accurate as we vary the rates of the process.

Further, the discrepancies in numerical transition probabilities between methods indeed translate to differences in estimated rates. To see this, we generate a partially observed dataset and infer rates using both methods. We simulate from the BDS process with parameters $\lambda = 0.07, \mu = 0.12, \nu = 0.02$ to resemble the dynamics of the real dataset we will analyze in the next section, and record 200 discretely observed states of the process evenly spaced $dt$ time units apart. Each simulated interval begins with an initial population size drawn uniformly between $1$ and $15$, and this data generating process is repeated three times, producing three datasets corresponding to inter-observation intervals of lengths $dt = (0.2,0.4,0.6)$. 
We infer the MLE rates for each of the three discretely observed datasets using the generating function method and under the frequent modeling assumption. This entire procedure is then repeated over 200 trials. 

In the top row of Figure \ref{fig:bias}, we see that our generating function approach successfully recovers the MLE estimates, and coverage of $95\%$ confidence intervals remains close to 0.95 as we increase the length of time intervals between observations. The FM method performs  somewhat reasonably for shorter observation intervals, but the bias in these approximate MLEs becomes stark as $dt$ increases, with $95\%$ confidence interval coverage probability dropping as low as $0.24$.
\begin{figure}[t]
\centering
\noindent \includegraphics[width=0.85\textwidth]{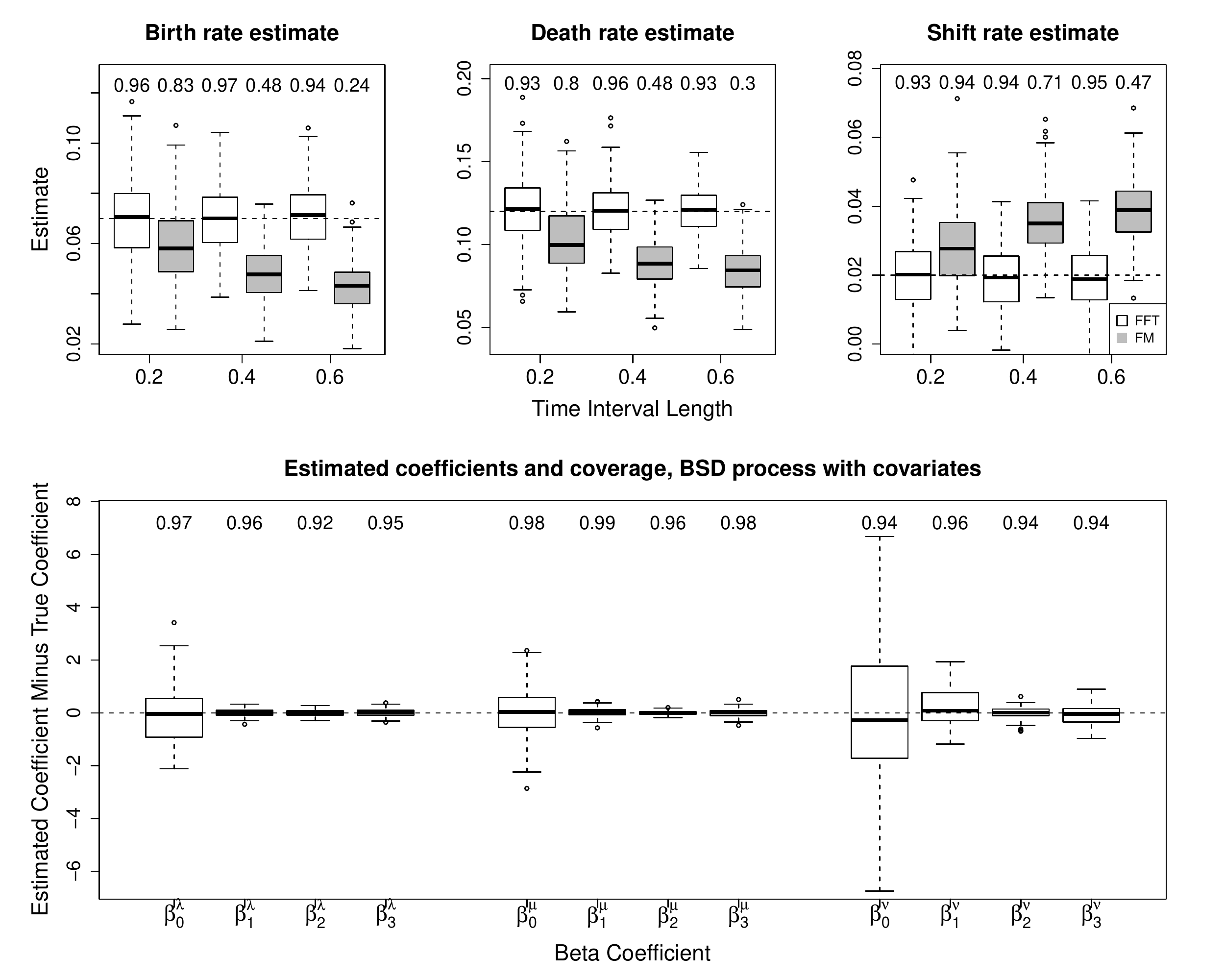}
\vspace{-0.4cm}
\caption{\small MLE parameter estimates on simulated data. The top row displays estimates of global birth, death, and shift rates in the simple BDS for three datasets, each with observation interval lengths $dt = (0.2,0.4,0.6)$. True parameter values used to initialize simulations marked by horizontal dashed line, and results using the FM method are included in gray. Monte Carlo coverage probabilities for $95\%$ confidence intervals are displayed above box plots. The bottom row displays estimated coefficients using EM in the BDS process with covariates, shifted by true values. }
\label{fig:bias}
\end{figure}

Similarly to this transition probability experiment, we check the accuracy of restricted moment computations via simulation by verifying the equality
\[ \text{E}(N_t^+ \mid X_0 = i,j) = \sum_{k,l} \text{E}(N_t^+ , 1_{x_t = kl} \mid x_0 = i,j) \]
for expected births, and analogous expressions for other expected sufficient statistics. The left hand side is empirically approximated by a Monte Carlo average of the number of births over many realizations of the process, while the restricted moments on the right hand side are the quantities computed via our generating function approach (see Appendix C, Figure C-2).

\subsection{Estimation of parameters in BDS model with covariates}

With accurate transition probabilities and restricted moments in hand, we are ready to infer coefficients in the BDS model with covariate-dependent rates using the EM algorithm. We again begin by generating simulated data to resemble the  dataset analyzed in the next section. The simulated dataset consists of observations corresponding to 100 ``patients", each with three covariates  $z_{p,1}, z_{p,2}, z_{p,3} \sim \text{Unif} \{ (0,2) \times (6,10) \times (4,6) \}$. We then simulate patient-specific BDS processes, beginning with rates $\lambda_p, \nu_p, \mu_p$ log-linearly related to a true vector of coefficients $\boldsymbol\beta$ as defined in \eqref{eq:loglin}. We collect between 2 and 7 observations per patient, each spaced $dt=0.4$ apart. Each simulated observation interval begins with an initial number of particles uniformly drawn between $2$ and $14$. Finally, we set true values of the effect sizes:
 $\boldsymbol{\beta}^\lambda = \left[ \log(7.5), \log(0.5), \log(0.3), \log(3) \right]$,
 $\boldsymbol{\beta}^\nu = \left[ \log(0.5), \log(8), \log(0.5), \log(0.9) \right]$,
 $\boldsymbol{\beta}^{\mu} = \left[ \log(4), \log(0.3), \log(0.8), \log(0.9) \right]$,
chosen so that averaging over patients, the overall birth, shift, and death rates of the process are similar to previous studies \citep{rosenberg2003, doss2013}.


 The EM algorithm is initialized with $\boldsymbol\beta_0 \sim N(\boldsymbol\beta, \text{diag}(0.5 \boldsymbol\beta))$, and the entire procedure of generating the dataset and inferring rates via EM is repeated 150 times. In the bottom row of Figure \ref{fig:bias}, we see that the MLEs are again unbiased estimates of the true values, with corresponding confidence interval coverage staying close to $95\%$.

Having verified that our EM algorithm successfully recovers the true parameters, we turn to a performance comparison with generic optimization via the Nelder-Mead (NM) algorithm implemented in the \texttt{optim} package \citep{nelder1965}. We choose NM as the method for comparison as it proved to be the most robust among the methods available via the \texttt{optim} function in R; a similar choice of NM for comparison to EM implementations is motivated in \citep{lange2013}. In this experiment, we generate one dataset as described in the procedure above from the BDS model with covariates. Fixing these data, we initialize each method with identical initial parameter values and convergence criteria, using a relative tolerance of $\epsilon = 1 \times 10^{-6}$, and repeat this procedure over $100$ sets of initial conditions.
\begin{figure}
    \centering
    \subfloat{{\includegraphics[width=.33\paperwidth]{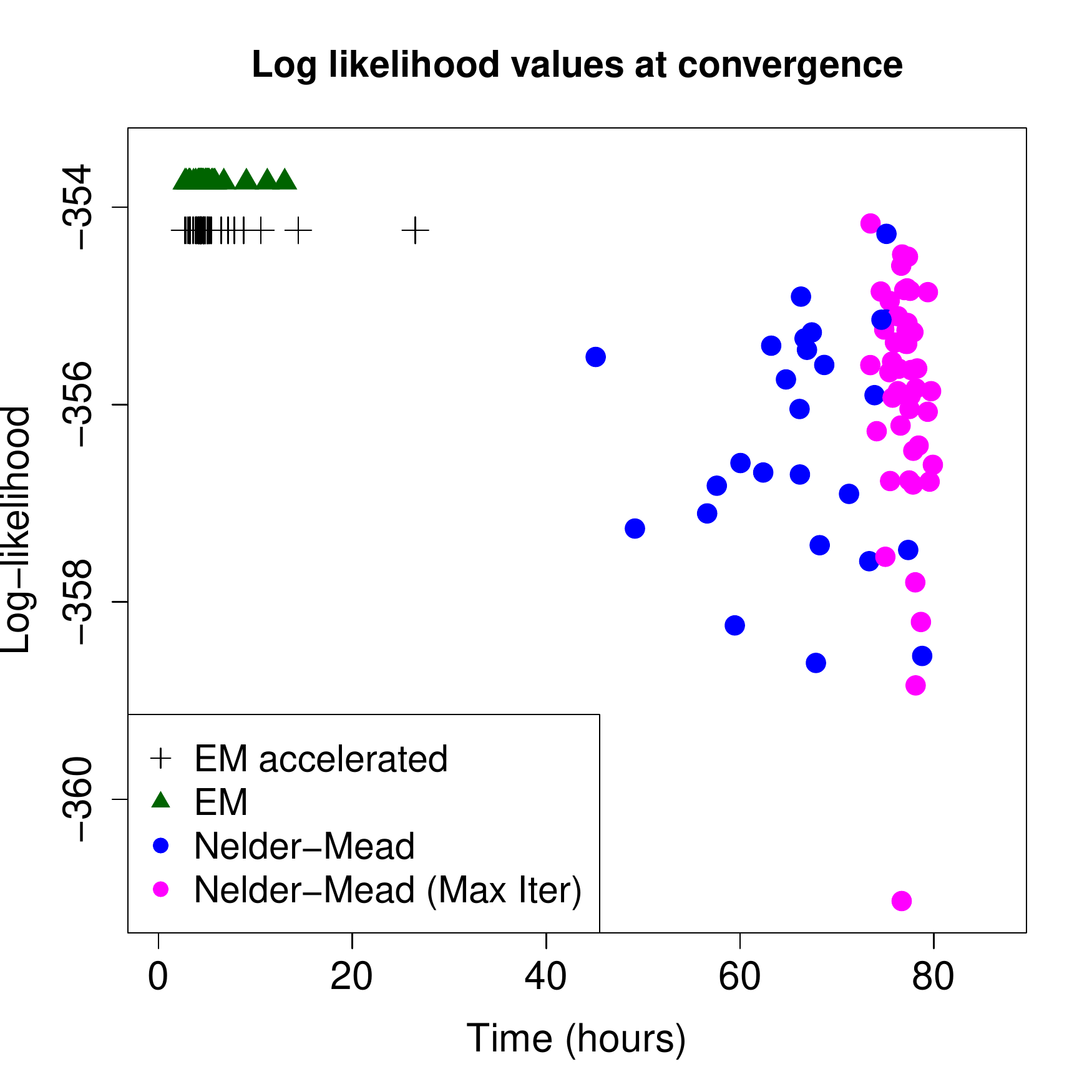} }}
    \subfloat{{\includegraphics[width=.33\paperwidth]{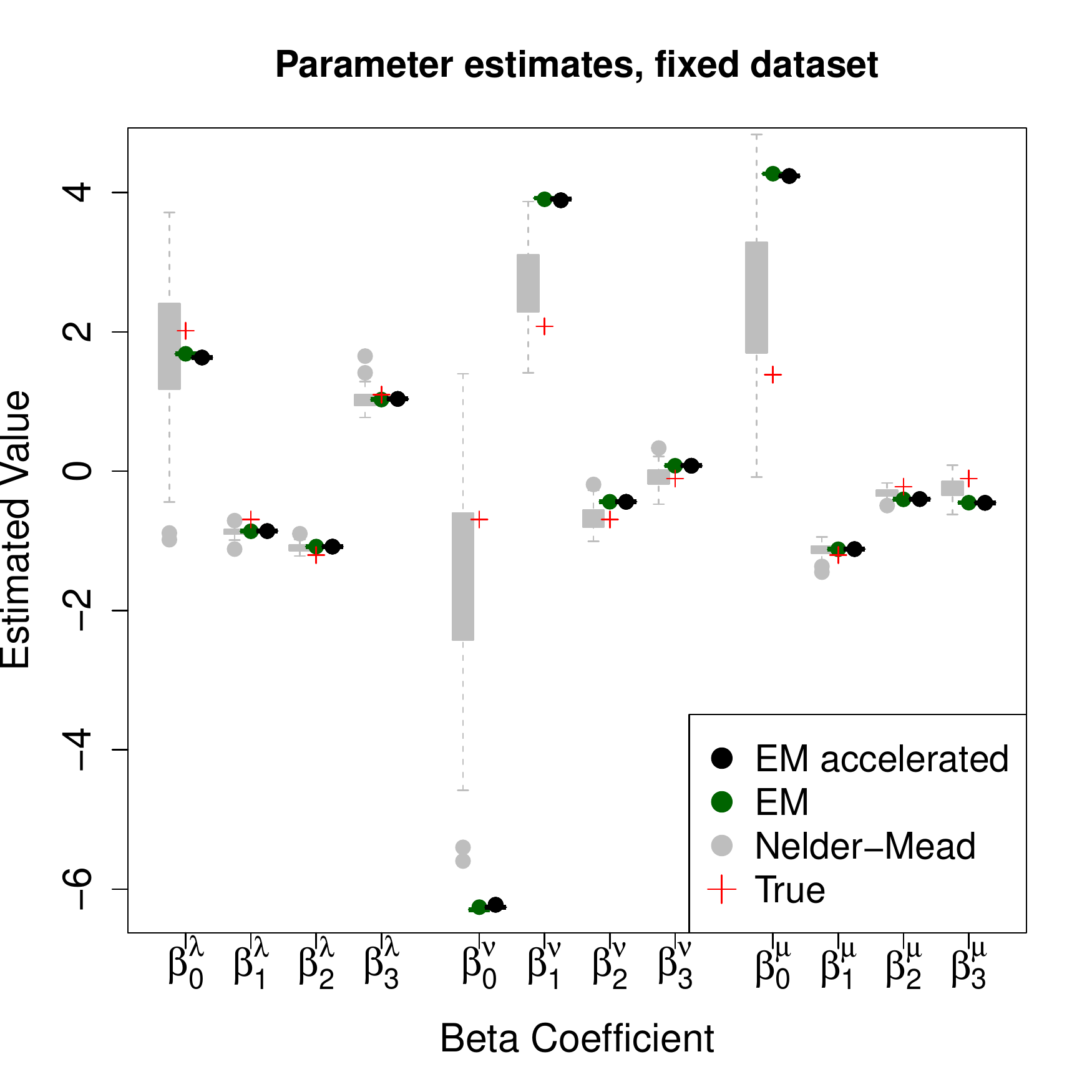} }}
    \caption{\small The left plot shows converged log-likelihood values using EM, accelerated EM, and Nelder-Mead optimization. The right plot shows parameter estimates produced by the EM, accelerated EM, and Nelder-Mead algorithms, with true parameters values shown as red crosses.}%
    \label{fig:loglike}
\end{figure}

Figure \ref{fig:loglike} displays the log-likelihood values achieved by each algorithm at convergence, as well as values in which Nelder-Mead terminated at an iteration limit set at 2000 steps. We see that in every case, the EM algorithm is significantly faster and finds a better optimum than NM. Further, the wide range of converged log-likelihood values suggests that NM is sensitive to initial conditions --- an undesirable feature in this fixed data setting. 
We also verify that accelerating our EM algorithm according to Section \ref{sec:efficientEM} does not affect the log-likelihood value at convergence up to numerical precision, with a total difference in log-likelihood less than $0.5$ accumulated over more than 400 observation intervals. Finally, we note that the comparison between EM and accelerated EM here is included to illustrate that they arrive at the same log-likelihood value and estimates at convergence. The increase in efficiency is not seen here: in these simulated examples, the generating function computations are always performed and cached at each iteration, rather than bypassed for candidate intervals described in  Section \ref{sec:efficientEM}. In our application to the real dataset in the next section, we find that accelerating EM runs approximately six times as fast as its nonaccelerated counterpart.
\par
Our EM approach is not only more stable in terms of the maximized log-likelihood, but also in terms of parameter estimates. The right panel of Figure \ref{fig:loglike} shows that estimates for each coefficient differ by no more than $0.01$ across disparate initial 
conditions under both EM implementations, while a range of estimates are produced by the Nelder-Mead algorithm. 
\par
Notice that for some coefficients, estimates produced by NM appear to lie closer to the ``true" parameters used to generate the synthetic data. We believe this to be an artifact of centering initial parameter values for both algorithms around the true parameters.  Indeed, MLEs corresponding to the likelihood surface of a given \textit{fixed} dataset generally do not exactly coincide with the ``true" parameters used to simulate the data. The fact that EM consistently finds a better optimum in terms of log-likelihood demonstrates that this is the case.


\subsection{\textit{Mycobacterium tuberculosis} transposable element evolution}
We apply our EM algorithm to infer covariate-dependent birth, death, and shift rates of the \textit{M. tuberculosis} transposon IS\textit{6110}, a frequently used marker to track \textit{M. tuberculosis} in the community \citep{mcevoy2007}. 
The marker serves as a DNA fingerprint, and in community-based studies patients that share the same or similar \textit{M. tuberculosis} genotypes are considered as part of the same transmission chain
\citep{van1993, kato2011}. However, such inference relies on a fairly precise understanding of within-host evolutionary dynamics: for instance, if a DNA marker changes very rapidly, isolates from the same source will be strongly differentiated, and the severity of outbreaks would be underestimated without accounting for the high change rate. Understanding the rates of change of IS\textit{6110}-based genotypes is thus critical toward the interpretation and design of such studies \citep{tanaka2001}, which in turn provide important information toward designing policy decisions such as control and intervention programs.

\begin{figure}[t]
    \centering
    \subfloat{{\includegraphics[width=.35\paperwidth]{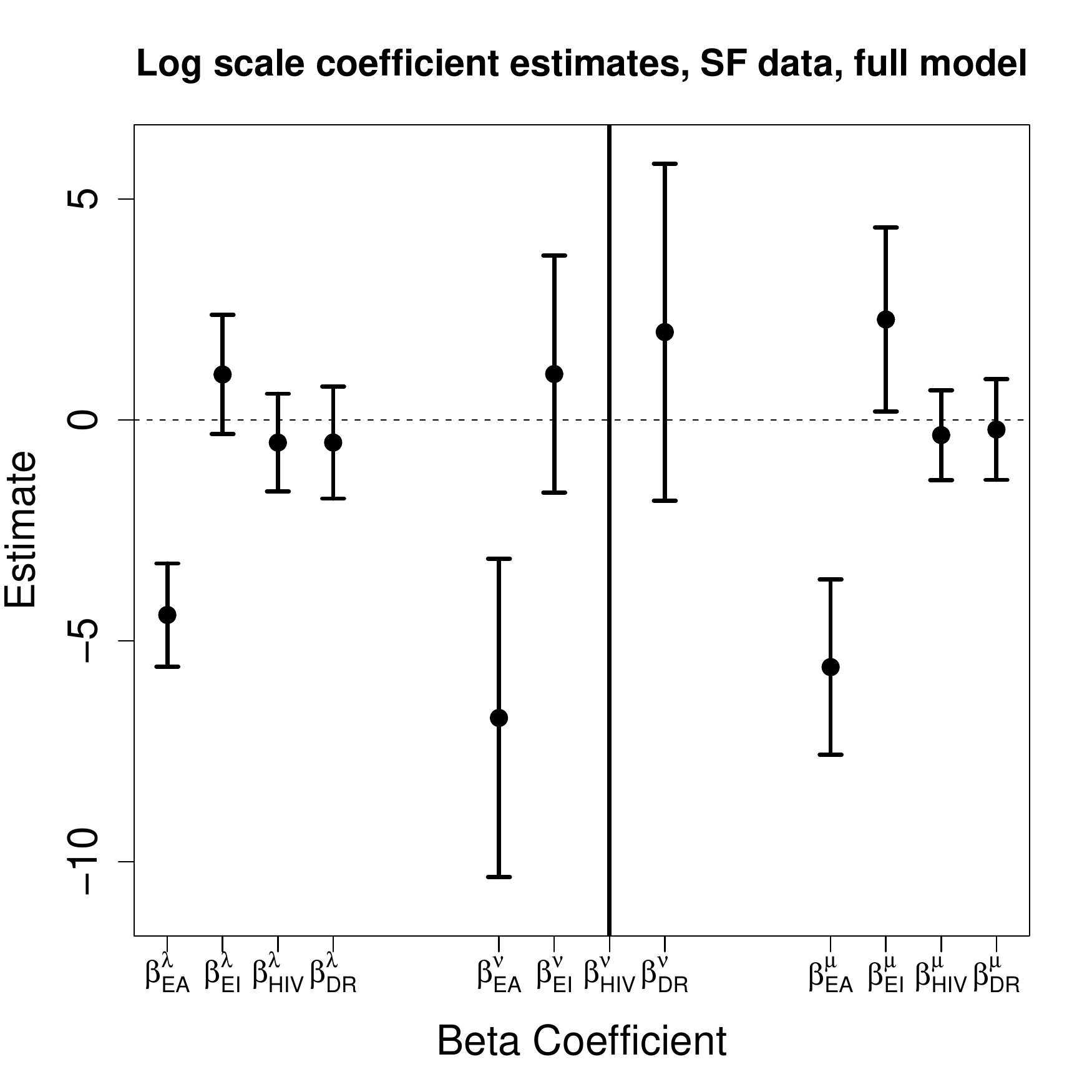} }}%
    \subfloat{{\includegraphics[width=.35\paperwidth]{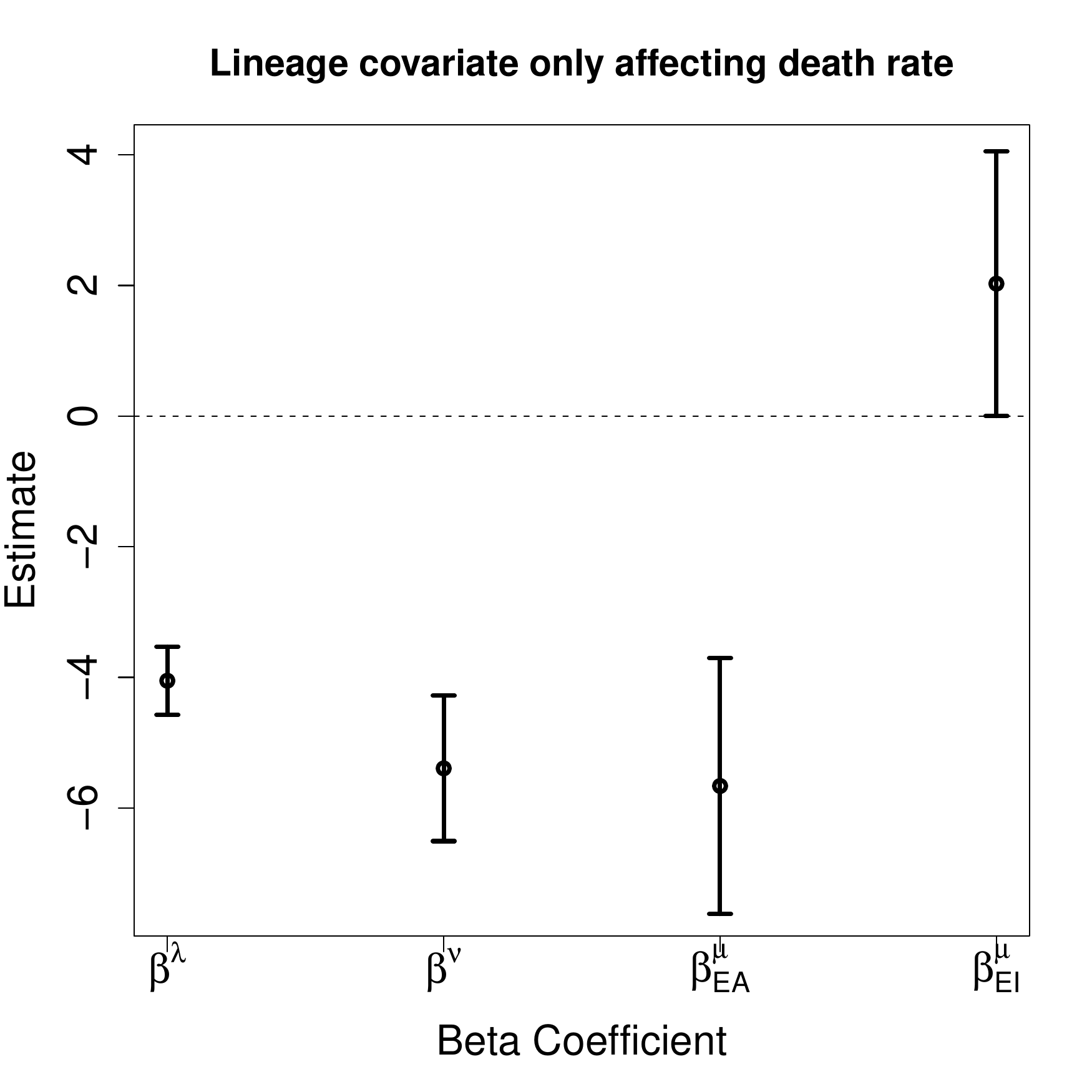} }}%
    \caption{\small Coefficient estimates and $95\%$ confidence intervals in full model and best model according to BIC. Notice intervals corresponding to $\beta_{EI}^\mu$ do not contain $0$.}%
    \label{fig:SF}%
\end{figure}

We analyze data from an ongoing study of the transmission and pathogenesis of \textit{M. tuberculosis} patients in a community study in San Francisco \citep{cattamanchi2006, suwanpimolkul2013}. The database includes all culture positive tuberculosis cases reported to the San Francisco Department of Public Health.  We included patients with more than one \textit{M. tuberculosis} isolate from specimens sampled more than 10 days apart, genotyped with IS\textit{6110} restriction fragment length polymorphism (RFLP) analysis. We assume that changes in the bands marking RFLP patterns evolve according to a linear birth-death-shift process, and assume that patients are not reinfected with a new strain between observations. Our dataset contains 252 observation intervals corresponding to 196 unique patients
observed at 452 time points. Average time between sampling times is $0.35$ years, with the longest interval being $2.35$ years. Of the $252$ intervals, $29$ feature end points with distinct genotypes.
\par
This dataset was analyzed by \cite{rosenberg2003} under the FM assumption, but these authors necessarily discarded all intervals with more than one change in RFLP bands, as these intervals with ``complex changes" are not possible under their restricted model. A later investigation by \cite{doss2013} relaxes this assumption, allowing for multiple births or deaths to occur, but ignores RFLP band locations entirely, working instead only with total copy numbers evolving under a linear birth-death process. 
Under this birth-death model,  the shift rate becomes unidentifiable, and the study instead infers covariate effects of birth and death rates. Our new method allows for a more principled, complete analysis, utilizing the full dataset without compromising any original modeling assumptions.
\par
We begin by applying our EM algorithm to the simple BDS model with a single birth, death, and shift rate of IS\textit{6110} for all patients. We estimate the MLE rates $\hat\lambda = 0.0156$, $\hat\nu = 0.00426$, $\hat\mu = 0.0187$, with associated $95\%$ confidence intervals $(0.00929, 0.0251)$, $(0.00145, 0.0125)$, and $(0.0177,0.0301)$ respectively. Starting the algorithm from a range of initial parameter values did not affect these results. These estimates are interpretable as the change rate of IS\textit{6110} per copy, per year, and our results are consistent with previous estimates in the literature: for all rates, confidence intervals overlap those obtained in the frequent monitoring approach in \citep{rosenberg2003} as well as those obtained in the BD model \citep{doss2013}. Similarly to \cite{doss2013} which estimates $\mu = 0.0207$, we find that our estimate of death rate $\mu$ is higher when allowing for multiple events between observations, compared to $\mu = 0.0147$ obtained under the FM assumption. This is to be expected, as there are three intervals in which IS\textit{6110} count drops by more than 1 in the dataset. Although confidence intervals overlap, our estimate of the shift rate is noticeably higher than the previous finding $\nu = 0.00268$ under FM, with the upper end of our confidence interval almost twice as large as the upper end of the $95\%$ FM confidence interval $[0, 0.00654)$. Again, our analysis allows inclusion of several intervals that can be explained by at least two genotype changes that were either omitted in earlier studies or interpreted as a single birth event. Our EM algorithm approach is the first method to our knowledge that is able to accurately estimate the shift rate and produce reliable confidence intervals in the BDS model.
\par
In addition to estimating the BDS rates globally, \citet{doss2013} investigated rates as functions of several covariates in a panel data setting, and their findings in the birth-death framework suggest that \textit{M. tuberculosis} lineage \citep{gagneux2006} may have a statistically significant effect on the rates of the process. 
We reexamine the effect of lineage on the rates in the full BDS model, considering 109 patients infected with Euro-American (EU) lineage strains, 54 patients with East-Asian (EA) strains, and 25 patients with Indo-Oceanic (IO) strains. We combine EU and IO lineages, because \cite{doss2013} found that the number of IO samples was not sufficient to recover rates for this
lineage.
Following \cite{doss2013}, we also include HIV infection status of each patient (HIV) and drug resistance status of the \textit{M. Tuberculosis} strain (DR). These attributes are coded as binary covariates: $\text{EI}_p = 1$ if patient $p$ is infected with the EU or IO strain and 0 otherwise, 
so that intercept terms $\beta^\lambda_0, \beta^\mu_0, \beta^\nu_0$ correspond to the EA strain. The variable $\text{HIV}_p = 1$ if patient $p$ is infected with HIV and 0 otherwise, and $\text{DR}_p= 1$ if patient $p$ is infected with a drug-resistant strain, and 0 otherwise.
Covariates are log-linearly related to birth, death, and shift rates:
$\log \lambda_p = \beta^\lambda_0 + \beta^\lambda_1 \text{EI}_p  + \beta^\lambda_2 \text{HIV}_p + \beta^\lambda_3 \text{DR}_p$,
$\log \mu_p = \beta^\mu_0 + \beta^\mu_1 \text{EI}_p  + \beta^\mu_2 \text{HIV}_p + \beta^\mu_3 \text{DR}_p$,
$\log \nu_p = \beta^\nu_0 + \beta^\nu_1 \text{EI}_p  + \beta^\nu_2 \text{HIV}_p + \beta^\nu_3 \text{DR}_p$.  

We estimate coefficients in the full log-linear model described above, as well as in several simpler models, using the EM algorithm. The simpler models differ from the full model by either excluding the HIV and DR covariates, or excluding all covariates for specified global or ``simple" rates. For instance, the model labeled ``Lineage only, simple $\nu$" in Table \ref{tab:BIC} has five parameters $\boldsymbol\beta = ( \beta^\lambda_0, \beta^\lambda_1,  \beta^\mu_0, \beta^\mu_1, \beta^\nu)$, and rates defined as
\[ \log \lambda_p = \beta^\lambda_0 + \beta^\lambda_1 \text{EI}_p, \hspace{15pt} \log \nu_p = \log \nu = \beta^\nu, \hspace{15pt} \log \mu_p = \beta^\mu_0 + \beta^\mu_1 \text{EI}_p. \]
In all cases, estimates obtained using the accelerated EM algorithm and regular implementation coincide, and neither is sensitive to initial conditions.
A summary and model comparison via the Bayesian Information Criterion (BIC) \citep{schwarz1978} is included in Table \ref{tab:BIC}, which selects the model including only the lineage covariate for modeling death rate $\mu$.  Coefficient estimates are displayed graphically for the full model as well as the best model selected by BIC in Figure \ref{fig:SF}. While we choose not to report coefficient estimates from each model for brevity, in \textit{all} models, the confidence interval for $\beta_{EI}^\mu$ does not contain zero, indicating that strain lineage has a statistically significant effect on the death rate. The estimate $\hat{\beta_{EI}^\mu} = 2.028$ under the best model indicates that in Euro-American and Indo-Oceanic lineages 
loss of IS\textit{6110} element occurs $\exp(2.028) =7.599$ times faster than in their East-Asian counterpart.
Our analysis affirms the result suggested by \cite{doss2013} in the simpler BD framework: \textit{M. tuberculosis} lineage needs to be taken into consideration when studying disease transmission using IS\textit{6110} genotypes.

\begin{table}
\begin{center}
\begin{tabular}{rrrr}
 \hline
Model & \# Params &  Log-likelihood & BIC \\ 
\hline
Full, separate EU, IO lineages & 15 & -119.845 & 330.01 \\
Full & 12 & -120.498 & 313.25 \\
Full, simple $\nu$ & 9 & -122.455 & 299.10 \\
Lineage covariate only & 6 &  -123.649 & 293.42 \\
Lineage only, simple $\nu$ & 5 & -123.717 & 277.54 \\
\textbf{Lineage only, simple $\boldsymbol\lambda, \boldsymbol\nu$} & \textbf{4} & \textbf{-124.472} & \textbf{273.02} \\
Simple $\lambda, \nu, \mu$ & 3 & -127.914 & 273.90 \\
\hline
\end{tabular}
\caption{ \small Model comparison via BIC $ \approx -2 \hat{L} + k  \ln n $. We also fit the log-linear model in \citet{doss2013}, which includes separate indicator variables for Euro-American and Indo-Oceanic lineages. Models described as ``lineage only" do not include HIV, DR covariates, and rates described as ``simple" are global to all patients, not influenced by covariates in the model. }
\label{tab:BIC}
\end{center}
\end{table}
\section{Discussion}

In this paper, we have developed an EM algorithm for inference in a discretely observed, multi-type branching process framework. We focus our attention on fitting BDS processes to panel data, driven by the problem of estimating evolutionary dynamics of IS\textit{6110} --- a genetic marker that plays an important role in DNA fingerprinting of \textit{M.\ tuberculosis}. Our method allows for birth, death, and shift log-rates to be linear combinations of many patient-specific covariates, and is flexible enough to capture the full range of dynamics between observation times by approximating the BDS process with a two-type branching process. To our knowledge, there is no other method of comparable accuracy for fitting BDS processes in this setting. 
\par
The generating functions we derive and numerical techniques that use these functions to calculate previously unavailable transition probabilities and restricted moments are helpful tools toward probabilistic characterization of such processes more generally. We demonstrate how our generating function approach leads to maximum likelihood estimation and evaluation of expected complete-data log-likelihood within an EM algorithm, but note that these calculations also arise in a variety of other statistical techniques for prediction and estimation. For example, availability and tractability of the likelihood via our methods allows for their use in Bayesian inference.
\par
Several problems associated with our numerical methods remain open. First, although we have empirical
evidence that our branching process approximation to the discretely observed BDS likelihood is very accurate, rigorous characterization of this approximation is 
lacking. Filling this theoretical gap is an interesting avenue for future research. Second, our method has potential numerical limitations in settings with high population counts.  Computing transition probabilities to population sizes up to $N$ of any particle type typically requires $N^p$ differential equations to be solved, where $p$ is the number of particle types. Although efficient numerical solvers are available and each ODE evaluation can be accomplished in only fractions of a second, requiring millions of evaluations becomes prohibitive, especially within an iterative algorithm. However, because the support of transition probabilities is often concentrated unless observation intervals are very long, future work may harness this sparsity to accelerate computations.
\par
We apply our method to analyze within-host evolution of the transposon IS\textit{6110}, an important marker in genetic fingerprinting of \textit{M. tuberculosis}. We obtain confidence intervals for global birth, death, and shift rates $\lambda, \mu, \nu$ that overlap with those obtained by \cite{rosenberg2003} and with those for $\lambda$ and $\mu$ obtained by \cite{doss2013}. Thus, our estimates are consistent with previous results.  While this suggests that the restrictive model assumptions in earlier approaches are not unreasonable in this application, we draw attention to our significantly higher estimate of $\nu$. Indeed, while the frequent monitoring study by \cite{rosenberg2003} suggests that the global shift rate $\nu$ is an order of magnitude smaller than the birth and death rates, our method reveals that the shift rate is in fact comparable to the baseline Euro-American death rate after accounting for strain lineage. This is clearly illustrated in Figure \ref{fig:SF}, and suggests that the non-negligible shift event should not be omitted from the model as it was in \citep{doss2013}. This novel observation was not possible using existing methodology --- our approach is the first to accurately estimate the birth, death, and shift rates as functions of covariates in this discretely monitored setting without compromising model assumptions.

Our covariate-specific rate analysis reaffirms previous indication in the simplified BD framework  that strain lineage has a significant effect on the death rate \citep{doss2013}, although the large confidence intervals suggest that this lineage effect is somewhat marginal. Indeed, more data would be required to be certain in the result, but our principled analysis is assuring in that any spurious result can now be attributed to limited, noisy data rather than to model misspecification. The possibility of differences in rates of genetic marker evolution across lineages is important in epidemiological studies. For example,  similar IS\textit{6110} genotypes across multiple individuals infected with EA lineage of \textit{M.\ tuberculosis} do not provide strong evidence of these individuals belonging to the same transmission chain, because of the slow change rate of IS\textit{6110} in the EA lineage. Failing to account for this may lead to inferring false relationships among genotypically similar clusters of patients. 
\par
The BDS model we consider is general enough so that our methods can be applied to studying evolution of any transposable element. Such studies are not limited to infectious disease surveillance, because studying evolution of transposable elements in
eukaryotes is also of great interest \citep{Biemont2010}.  Beyond the BDS framework,
the tools we develop for fitting branching processes are transferable to many settings. For example, our methodology is applicable to compartmental models, a class of well-known multi-type branching processes that finds applications in modeling cancerous growth, bacterial evolution, and cellular differentiation in systems such as hematopoiesis  \citep{gibson1998, golinelli2006}. 
%



{\small \bibliography{bds_paper}}

\setcounter{table}{0}
\renewcommand{\thetable}{C-\arabic{table}}
\renewcommand{\thefigure}{C-\arabic{figure}}

\renewcommand{\theequation}{A-\arabic{equation}}
\setcounter{equation}{0}

\section*{Appendix A} Here we derive and solve the Kolmogorov backward equations of the two-type branching process necessary for evaluating the generating functions whose coefficients yield transition probabilities. See \citep{bailey1990} for an exposition on this solution technique.

Our two-type branching process is represent by a vector $(X_1(t), X_2(t))$ that denotes the numbers of particles of two types at time $t$. Recall the quantities $a_1(k,l)$, the rates of producing $k$ type 1 particles and $l$ type 2 particles, starting with one type 1 particle, and $a_2(k,l)$, analogously defined but beginning with one type 2 particle. Then we may introduce respective pseudo-generating functions
$u_i(s_1,s_2) = \sum_k \sum_l a_i(k,l)s_1^k s_2^l$ for $i=1,2$, and the probability generating functions can be expressed
\begin{align}
\phi_{10}(t, s_1, s_2) &= E \left[s_1^{X_1(t)} s_2^{X_2(t)} \mid X_1(0) = 1, X_2(0) = 0\right] = \sum_{k=0}^\infty \sum_{l=0}^\infty P_{(1,0), (k,l)} (t) s_1^k s_2^l \nonumber \\
&= \sum_{k=0}^\infty \sum_{l=0}^\infty \left[ \mathbf{1}_{k=1, l = 0} + a_1(k,l) t + o(t) \right] s_1^k s_2^l = s_1 + u_1(s_1, s_2) t + o(t). \label{eq:pgf}
\end{align}
An analogous expression for $\phi_{01}(t, s_1, s_2)$ is obtained similarly. For short, we write $\phi_{10} := \phi_1, \phi_{01} := \phi_2$, and thus we have the following relations between $\phi$ and $u$
\[\frac{d \phi_1 (t, s_1, s_2)}{dt} \Big|_{t=0} = u_1(s_1, s_2), \hspace{20pt} \frac{d \phi_2 (t, s_1, s_2)}{dt} \Big|_{t=0} = u_2(s_1, s_2). \]
By particle independence, $\phi_{i,j} = \phi_1^i \phi_2^j$, so it suffices to work with only $\phi_1, \phi_2$.
We now derive the backward equations for $\phi_1$ and $\phi_2$. Chapman-Kolmogorov equations yield the symmetric relations
\begin{align}
\phi_1(t+h, s_1, s_2) &= \phi_1(t, \phi_1(h, s_1, s_2), \phi_2(h, s_1, s_2)) \label{eq:chap1}
\\ &= \phi_1(h, \phi_1(t, s_1, s_2), \phi_2(t, s_1, s_2)). \label{eq:chap2}
\end{align}
To derive the backward equations, we begin by expanding $\phi_1(t+h, s_1, s_2)$ around $t$ and applying \eqref{eq:chap2}:
\begin{align*}
\phi_1(t+h, s_1, s_2) &= \phi_1(t, s_1, s_2) + \frac{d \phi_1 (t+h, s_1, s_2)}{dh} \Big|_{h=0} h + o(h) 
\\ &= \phi_1(t, s_1, s_2) + \frac{d \phi_1(h, \phi_1(t, s_1, s_2), \phi_2(t, s_1, s_2))}{dh} \Big|_{h=0} h + o(h) 
\\ &= \phi_1(t,s_1,s_2) + u_1( \phi_1(t,s_1, s_2) , \phi_2(t, s_1, s_2)) h + o(h).
\end{align*}

Since an analogous argument applies for $\phi_2$, we arrive at the system
\[ \begin{cases}
\frac{d}{dt} \phi_1(t, s_1, s_2) = u_1( \phi_1(t, s_1, s_2), \phi_2(t, s_1, s_2) ), \\
\frac{d}{dt} \phi_2(t, s_1, s_2) = u_2( \phi_1(t, s_1, s_2), \phi_2(t, s_1, s_2) ), 
\end{cases} \]
subject to initial conditions $\phi_1(0, s_1, s_2) = s_1, \phi_2(0, s_1, s_2) = s_2$. 

We now substitute the rates specific to our birth-shift-death model into this general form: recall the rates defining the two-type branching process formulation presented in Section 2.4 of the main paper are
\begin{align}
a_1(1,1) &= \lambda, && a_1(0,1) = \nu, && a_1(0,0) = \mu && a_1(1,0) = -(\lambda + \nu + \mu),  \nonumber \\
a_2(0,2) &= \lambda, && a_2(0,1) = -(\lambda + \mu), && a_2(0,0) = \mu, \label{eq:rates}
\end{align}
so that the pseudo-generating functions and backward equations are 
\begin{equation}\label{eq:backward}
\begin{cases}
 u_1(s_1,s_2) = \lambda s_1 s_2 + \nu s_2 + \mu - (\lambda + \nu + \mu)s_1, & \hspace{20pt} \frac{d}{dt} \phi_1 = \lambda \phi_1 \phi_2 + \nu  \phi_2 + \mu - (\lambda + \nu + \mu) s_1,\\
u_2(s_1,s_2) = \lambda s_2^2 - (\lambda + \mu) s_2 + \mu,    &  \hspace{20pt} \frac{d}{dt} \phi_2 = \lambda \phi_2^2 - (\lambda + \mu) \phi_2 + \mu. 
\end{cases}
\end{equation}
Upon rearranging, the expression for $\phi_2$ becomes a Ricatti equation
$$\phi_2 ' - \lambda \phi_2^2 + (\lambda + \mu) \phi_2 = \mu,  $$
and the constant solutions $\phi_2 = 1, \mu/ \lambda $ are both particular solutions. Using the simpler root $\phi_2 = 1$, we can reduce the above Ricatti equation to a linear ODE by making a substitution $z = \frac{1}{\phi_2 -1}$, so that $\phi_2 = 1 + \frac{1}{z}$:
\begin{align*}
 \phi_2' &= -\frac{z'}{z^2} = \mu - (\lambda + \mu)(\frac{1}{z}+1) + \lambda(1 + \frac{1}{z})^2 &= \mu - \frac{\lambda + \mu}{z} - (\lambda + \mu) + \lambda( \frac{1}{z^2} + \frac{2}{z} + 1)
 \\ &= -\frac{\mu - \lambda}{z} + \frac{\lambda}{z^2}.
\end{align*}
Multiplying through by $-z^2$ and rearranging, we arrive at a linear equation that is easily solved via the integrating factor method:
$$ z' + (\lambda - \mu)z = -\lambda \Rightarrow z = -\frac{\lambda}{\lambda - \mu} + C e^{-(\lambda - \mu) t}. $$
Substituting $\phi_2$ back into the expression, we obtain
$$\phi_2 = 1 + \frac{1}{ \frac{\lambda}{\mu - \lambda} + Ce^{(\mu - \lambda)t}},$$
and plugging in the initial condition $\phi_2(0,s_1,s_2) = s_2$, we see $C = \frac{1}{s_2-1} + \frac{\lambda}{\lambda - \mu}$.
Thus, we arrive at the closed form solution
\begin{equation}
\phi_2(t,s_1,s_2) = 1 + \left[ \frac{\lambda}{\mu - \lambda} + (\frac{1}{s_2-1} + \frac{\lambda}{\lambda - \mu})e^{(\mu - \lambda)t} \right] ^{-1} := g(t,s_1, s_2)
\end{equation}
We can now plug this solution into the ODE for $\phi_1$ to obtain
\begin{equation}\label{eq:phi1}
\frac{d}{dt} \phi_1 + (\lambda + \nu + \mu - \lambda g) \phi_1 = \nu g + \mu .
\end{equation}

\subsection*{Closed form solution for $\phi_1$}
Equation \eqref{eq:phi1} is linear with variable coefficients, and can again be solved by multiplying by an integrating factor. If we define the integrating factor
%
$\psi := \exp\left[ \int{(\lambda + \nu + \mu - \lambda g) dt}\right]$, then 
$$\frac{d}{dt}(\phi_1 \psi) = \psi(\nu g + \mu), $$
and after integration and rearranging,
\begin{equation}\label{p1sol}
\phi_1 = \psi^{-1} \left[ \int{ \psi( \nu g + \mu ) dt} + C \right]. 
\end{equation}
After further simplification, we may write
$$\psi = e^{ (\nu + \mu) t} (\lambda s_2 - \mu) + \lambda e^{(\lambda + \nu)t}(1 - s_2), $$
and the integrand becomes
\begin{equation}\label{integrand}
\psi(\nu g + \mu) = (\nu + \mu) \psi + \frac{\nu \psi} { \frac{\lambda}{\mu - \lambda} + (\frac{1}{s_2-1} + \frac{\lambda}{\lambda - \mu} ) e^{(\mu - \lambda)t} } .
\end{equation}
Integrating \eqref{integrand} and plugging into \eqref{p1sol} with initial condition $\phi_1(0,s_1, s_2) = s_1$, we ultimately obtain a closed form expression

\begin{align}
\phi_1(t, s_1, s_2) &=    \left[  e^{(\nu  + \mu ) t} (\lambda  s_2 -\mu ) + \lambda  e^{(\lambda +\nu ) t} (1-s_2 ) \right]^{-1} \nonumber
\\ & \cdot \Bigg\{ \nu (\mu - \lambda) e^{\nu t} \biggl[ \frac{ e^{\mu t}( \lambda s_2 - \mu) {}_2 F_1 \Big( 1, \frac{\mu + \nu}{\mu - \lambda}, \frac{\lambda - 2\mu - \nu}{\lambda - \mu}, \frac{ e^{(\mu - \lambda)t}(\lambda s_2 - \mu)}{ \lambda(s_2 - 1)} \Big) }{\lambda(\mu+\nu)} \nonumber
\\ & + \frac{e^{\lambda t} (1 - s_2) {}_2 F_1\Big(1, \frac{\lambda + \nu}{\mu - \lambda}, \frac{\mu + \nu}{\mu - \lambda}, \frac{ e^{(\mu - \lambda)t}(\lambda s_2 - \mu)}{ \lambda(s_2 - 1)} \Big)}{\lambda + \nu} \biggl] + (\lambda  s_2 -\mu ) e^{(\mu +\nu ) t} \nonumber
\\& + \frac{\lambda  (\nu +\mu ) (1-s_2) e^{(\lambda +\nu ) t}}{\lambda +\nu } + \mu  + s_1  (\lambda -\mu ) - \lambda  s_2  + \frac{ \lambda  (s_2 -1) (\nu +\mu ) } {\lambda +\nu} \nonumber
\\ & + \nu(\lambda - \mu) \biggl[ \frac{\lambda s_2 - \mu}{\lambda (\mu + \nu)}  {}_2 F_1\Big( 1, \frac{\mu + \nu}{\mu - \lambda}, \frac{\lambda -2\mu - \nu}{\lambda - \mu}, \frac{\lambda s_2 - \mu}{\lambda(s_2-1)}\Big) \nonumber
\\ & + \frac{1 - s_2}{\lambda + \nu}  {}_2 F_1 \Big(1, \frac{\lambda + \nu}{\mu - \lambda}, \frac{\mu + \nu}{\mu - \lambda}, \frac{\lambda s_2 - \mu}{\lambda(s_2-1)} \Big) \biggl] \Bigg\}, \label{eq:closedform}
\end{align}
where ${}_2 F_1$ indicates the hypergeometric function. In practice, we solve for $\phi_1$ numerically rather than using this closed form solution: evaluating \eqref{eq:phi1} via Runge-Kutta methods proves more stable than numerical evaluation of the hypergeometric functions arising in \eqref{eq:closedform}.

\renewcommand{\theequation}{B-\arabic{equation}}
\setcounter{equation}{0}

\section*{Appendix B}
Here we derive the equations in our main theorem. The formulation is repeated below:

\begin{theorem}\label{thm:main}
Let $\{ X_t \}$ be a two-type branching defined by the rates in equation \eqref{eq:rates}. Denote particle time and the number of births, shifts, and deaths over the interval $[0,t)$ by $R_t, b_t, f_p$, and $d_t$ respectively. Define the generating functions corresponding to births as 
\begin{align*}
H^+_1(r,s_1, s_2, t) &= \text{E} \left[ r^{b_t}  s_1^{X_1(t)} s_2^{X_2(t)} \mid \mb{X}(0) = (1,0) \right] \text{ and}\\
H^+_2(r,s_1, s_2, t) &= \text{E} \left[ r^{b_t}  s_1^{X_1(t)} s_2^{X_2(t)} \mid \mb{X}(0) = (0,1) \right].
\end{align*}
Then
\[
H_2^+ = y_b + \left[  \frac{-\lambda r}{2 \lambda r y_b - \lambda - \mu} + \left( \frac{1}{s_2 - y_b} + \frac{\lambda r}{2 \lambda r y_b - \lambda - \mu} \right)  e^{-(2y_b \lambda r - \lambda - \mu)t} \right] ^{-1},
\]
where $y_b = (\lambda + \mu + \sqrt{ \lambda^2 + 2 \lambda \mu + \mu^2 - 4 \lambda \mu r})/(2 \lambda r)$, and $H_1^+$ satisfies the following differential equation:
\begin{equation}\label{eq:birthODE}
\frac{d}{dt} H_1^+(t,s_1,s_2,r) = \lambda r H_1^+ H_2^+ + \nu H_2^+ + \mu - (\lambda + \mu + \nu) H_1^+, 
\end{equation}
subject to initial condition $H_1(r,s_1,s_2,0) = s_1$.

The analogous generating functions for shifts, deaths, and particle time satisfy the following equations:
\begin{align*}
H_2^- (t,s_1,s_2,r) &= y_d + \left[  \frac{-\lambda }{2 \lambda y_d - \lambda - \mu} + \left( \frac{1}{s_2 - y_d} + \frac{\lambda }{2 \lambda y_d - \lambda - \mu} \right)  e^{-(2y_d \lambda  - \lambda - \mu)t} \right] ^{-1},\\
H_2^\rightarrow (t,s_1,s_2,r) &= 1 + \left[ \frac{\lambda}{\mu - \lambda} + (\frac{1}{s_2-1} + \frac{\lambda}{\lambda - \mu})e^{(\mu - \lambda)t} \right] ^{-1}, \\
H_2^*(t,s_1,s_2,r) &= y_* + \left[  \frac{-\lambda }{2 \lambda y_* - \lambda - \mu - r} + \left( \frac{1}{s_2 - y_*} + \frac{\lambda }{2 \lambda y_* - \lambda - \mu - r} \right)  e^{-(2y_* \lambda  - \lambda - \mu - r)t} \right] ^{-1} ,
\end{align*}
\begin{align*}
\frac{d}{dt} H_1^- (t,s_1,s_2,r) &=  \lambda H_1^- H_2^- + \nu H_2^- + \mu r - (\lambda + \mu + \nu) H_1^-, \\
\frac{d}{dt} H_1^\rightarrow (t,s_1,s_2,r) &=  \lambda H_1^\rightarrow H_2^\rightarrow + \nu r H_2^\rightarrow + \mu  - (\lambda + \mu + \nu) H_1^\rightarrow,  \\
\frac{d}{dt} H_1^* (t,s_1,s_2,r) &=  \lambda H_1^* H_2^* + \nu H_2^* + \mu - (\lambda + \mu + \nu + r) H_1^*,
\end{align*}
where $y_d = (\lambda + \mu + \sqrt{ \lambda^2 + 2 \lambda \mu + \mu^2 - 4 \lambda \mu r})/(2 \lambda)$, 
$y_* = (\lambda + \mu + r +  \sqrt{( \lambda + \mu + r)^2 - 4 \lambda \mu})/(2 \lambda)$ and 
$H_1^-(r,s_1,s_2,0) = H_1^\rightarrow(r,s_1,s_2,0) = H_1^*(r,s_1,s_2,0) = s_1$.
\end{theorem}

\begin{proof}
Begin by expanding 
\[ H_{10}^+ (t, r, s_1, s_2) = \sum_n \sum_k \sum_l Pr(b_t  = n, x_t = (k,l) | x_0 = (1,0) ) s_1^k s_2^l r^n .  \]
Recall the jump rates of the process in equation \eqref{eq:rates}: $a_1$ correspond to the process beginning with 1 type one particle, and $a_2$ are jump rates starting with 1 type two particle. 
We can express the probability terms in $H_{10}^+$ using the same type of first-order decomposition as in equation \eqref{eq:pgf}; for instance, in the event of a birth, \[Pr(b_t = 1, x_t = (1,1) | x_0 = (1,0) ) = a_1(1,1) + o(t)  = \lambda + o(t) \] 
and for other values of $n > 1$, \[ Pr(b_t = n, x_t = (1,1) | x_0 = (1,0) ) = o(t). \]
In the case of a shift,
\[ Pr(b_t = 0, x_t = (0,1) | x_0 = (1,0)) = a_1(0,1) + o(t) = \nu + o(t) \] 
and for other values of $n \neq 0$, \[ Pr(b_t = n, x_t = (0,1) | x_0 = (1,0) ) = o(t). \] 
We see that the $r^n$ term in the series $H_{10}^+$ is either $r^1 = r$ if exactly one birth occurs, or $r^0 = 1$ as other powers correspond to more than one event and are absorbed into the $o(t)$ term. Thus, 
\begin{align*}
H_{10}^+ (t, r, s_1, s_2) &= \sum_k \sum_l g_{10,kl}(r,t) s_1^k s_2^l = \sum_n \sum_k \sum_l Pr(b_t  = n, x_t = (k,l) | x_0 = (1,0) ) s_1^k s_2^l r^n \\
&=  s_1 + \lambda s_1 s_2 r + \nu s_2 + \mu - (\lambda + \nu + \mu) s_1 + o(t) := s_1 + u_1^b(s_1,s_2) t + o(t)
\end{align*}
with $u_1^b$ denoting the pseudo-generating function, similarly to \eqref{eq:pgf}. 
With an analogous derivation for $u_2^b$, we arrive at the system
\begin{equation}
\begin{cases}
u_1^b(s_1,s_2) = \lambda r s_1 s_2 + \nu s_2 + \mu - (\lambda + \nu + \mu)s_1\\
u_2^b(s_1,s_2) = \lambda r s_2^2 - (\lambda + \mu) s_2 + \mu, \label{eq:birthPseudo}
\end{cases}
\end{equation}
and since
\[ \frac{d H_{10}^+ (t,r,s_1,s_2)}{dt} \Big| _{t=0} = u_1^b(s1,s2,r), \hspace{25pt} \frac{d H_{01}^+ (t,r,s_1,s_2)}{dt} \Big| _{t=0} = u_2^b(s1,s2,r), \]
we obtain the backward equations system
\begin{equation}\label{eq:backwardBirths}
\begin{cases}
\frac{d}{dt} H_{10}^+(t, s_1, s_2,r) = u_1^b( H_{10}^+(t, s_1, s_2,r), H_{01}^+(t, s_1, s_2,r) ), \\
\frac{d}{dt} H_{01}^+(t, s_1, s_2,r) = u_2^b( H_{10}^+(t, s_1, s_2,r), H_{01}^+(t, s_1, s_2,r) ) 
\end{cases}
\end{equation}
by the same Chapman-Kolmogorov argument used for transition probabilities, subject to initial conditions $H_{10}(t = 0, s_1, s_2, r) = s_1$ and $H_{01}(t = 0, s_1, s_2, r) = s_2$. The systems for deaths and shifts are derived analogously beginning with this first-order expansion technique, and are respectively given by

\begin{equation}
\begin{cases}
u_1^d(s_1,s_2) = \lambda s_1 s_2 + \nu s_2 + r \mu - (\lambda + \nu + \mu)s_1\\
u_2^d(s_1,s_2) = \lambda s_2^2 - (\lambda + \mu) s_2 + r \mu, \label{eq:deathPseudo}
\end{cases}
\end{equation}

\begin{equation}
\begin{cases}
u_1^\rightarrow(s_1,s_2) = \lambda  s_1 s_2 + r \nu s_2 + \mu - (\lambda + \nu + \mu)s_1,\\
u_2^\rightarrow(s_1,s_2) = \lambda  s_2^2 - (\lambda + \mu) s_2 + \mu. \label{eq:shiftPseudo}
\end{cases}
\end{equation}

To derive the system governing the particle time generating function, recall the quantity $q^*_{ij,kl}(x;t) := Pr(R_t \leq x, X(t) = (k,l) | X(0)= (i,j) )$, and consider its Laplace-Stieltjes transform
\begin{equation}\label{eq:laplacestieltjes}
V_{ij,kl}(r;t) = \int_0^\infty e^{-rx} \mathrm{d} q^*_{ij,kl}(x;t) .
\end{equation}
The Laplace-Stieltjes transform of such a probability distribution corresponding to a \textit{reward function}, where $a_{ij}$ is the reward accrued per unit time spent in state $(i,j)$, satisfies the forward equation 
\begin{equation}\label{eq:forwardprob}
\frac{d}{dt} V_{ij,kl}(r;t) = -a_{ij} r V_{ij,kl}(r;t) + \sum_{m=1}^K \sum_{n=1}^K Q_{ij,mn} V_{ij,kl}(r;t) ,
\end{equation}
where $\mathbf{Q}$ is the infinitesimal generator of the Markov chain, with finite or countable number of rows and columns and entries $Q_{ij,kl}$ the instantaneous rates of transitioning from state $(i,j)$ to $(k,l)$, and $ Q_{ij,ij} = -\sum_{m,n \neq i,j} Q_{ij,mn}$. 
Following Neuts \citep{neuts1995}, we derive the following  integral equation:
\[ q_{ij,kl}^*(x,t) = \mathbf{1}_{\{ij=kl\}} \mathbf{1}_{\{x \geq a_{ij}t\}} e^{Q_{ij,ij}t} + \sum_{m,n \neq i,j} \int_0^t e^{Q_{ij,ij} u} Q_{ij,mn} q^*_{mn.kl}(x-a_{ij}u, t-u) \mathrm{d}u. \] 
Taking the Laplace transform of both sides and denoting $\widetilde{V}_{ij,kl}(r;t) = \int_0^\infty e^{-rx} q_{ij,kl}^*(x;t) \mathrm{d} x$, we obtain
\[ \widetilde{V}_{ij,kl}(r,t) = \mathbf{1}_{\{ij=kl\}}r^{-1} \exp \left[ (Q_{ij,ij}-a_{ij})t \right] + \sum_{m,n \neq i,j} \int_0^t e^{Q_{ij,ij}u} Q_{ij,mn} \mathrm{d}u \int_{a_{ij}u}^\infty e^{-rx}q^*_{mn,kl}(x-a_{ij}u; t-u) \mathrm{d}x. \]
Making a change of variables $y = x - a_{ij}u $ in the rightmost integral and multiplying both sides by $\exp \left[ - (Q_{ij,ij} - a_{ij})t \right]$ yields
\[ \exp \left[-(Q_{ij,ij} - a_{ij})t \right] \widetilde{V}_{ij,kl}(r,t) = \frac{1}{r} + \sum_{m,n \neq i,j} \int_0^t \exp \left[ -(Q_{ij,ij} - a_{ij}) (t-u) \right] Q_{ij,mn} \widetilde{V}_{mn,kl}(r; t-u). \]
Next, make another substitution $v = t-u$ and simplify after differentiating the above equation with respect to $t$: we arrive at
\[ \frac{\partial}{\partial t} \widetilde{V}_{ij,kl}(r;t) = -a_{ij} r \widetilde{V}_{ij,kl}(r;t) + \sum_{m=1}^K \sum_{n=1}^K Q_{ij,mn} \widetilde{V}_{mn,kl} (r;t) .\]
Equation \eqref{eq:forwardprob} then follows from $V_{ij,kl}(r;t) = s \widetilde{V}_{ij,kl}(r;t)$, with $V_{ij,kl}(t)(r;0)=~\mathbf{1}_{\{ij=kl\}}$.

The matrix $\mathbf{V}(r;t) := \left\{ V_{ij,kl}(r;t) \right\}$ can therefore be written as a matrix exponential
\begin{equation}\label{eq:matrixexp} 
\mathbf{V}(r;t) := \exp [ \mathbf{Q} - \text{diag}(\mb{a}) r) t] := \exp ( \widetilde{\mathbf{Q}} t),
\end{equation}
where $\text{diag}(\mb{a})$ is the diagonal matrix with diagonal entries $a_{ij}$.
In our case, $a_{ij} = 1$, since we are interested in particle time and the ``reward" that accumulates per unit of time is that quantity of time itself. 
Strictly speaking we don't need infinite dimensional matrix algebra here, but we use it to simplify our notation. 

Note the similarity of equation \eqref{eq:matrixexp} to the matrix exponential corresponding to transition probabilities $\mathbf{P}(t)= \exp(\mathbf{Q} t) $: thus, the system of backward equations for $V_{ij,kl}$ are almost identical to those for transition probabilities $p_{ij,kl}$. The generators $\widetilde{\mathbf{Q}} \neq \mathbf{Q}$ differ only in diagonal entries: instantaneous rates of no event occurring are augmented by an extra $r$ term $\widetilde{Q}_{ij,ij} = -\sum_{m,n \neq i,j} Q_{ij,mn} - r$. The system of backward equations is thus given by 
\begin{equation}\label{eq:PTsystem}
\begin{cases}
u_1^*(s_1,s_2) = \lambda s_1 s_2 + \nu s_2 + \mu - (\lambda + \nu + \mu + r) s_1,\\
u_2^*(s_1,s_2) = \lambda s_2^2 + \mu - (\lambda + \mu + r)s_2,
\end{cases}
\end{equation}
and as we have seen in the derivation for expected births in equation \eqref{eq:backwardBirths}, this implies that the generating function
\[ H_{10}^*(r,s_1,s_2,t) = \sum_k \sum_l \int_0^\infty e^{-rx} \mathrm{d} q^*_{ij,kl}(x;t) = \sum_k \sum_l V_{10,kl}(r,t) s_1^k s_2^l \]
also satisfies the same system. 

\subsection*{Reducing the systems}
Each of the four systems for births, shifts, deaths, and particle time can be reduced to a single ODE by first solving the second equation analytically. We demonstrate this in the case of the birth equations  \eqref{eq:backwardBirths}, and abbreviate $H_{10}^+ := H_1, H_{01}^+ := H_2$. Plugging \eqref{eq:birthPseudo} into \eqref{eq:backwardBirths},
\[\begin{cases}
\frac{d}{dt} H_1(t, s_1, s_2,r) = \lambda r H_1 H_2 + \nu H_2 + \mu - (\lambda + \nu + \mu)H_1, \\
\frac{d}{dt} H_2(t, s_1, s_2,r) =  \lambda r H_2^2 - (\lambda + \mu) H_2 + \mu. 
\end{cases} \]
The second equation is a Ricatti equation. To solve it, we first identify a constant solution 
\[y_b = \frac{ \lambda + \mu + \sqrt{ \lambda^2 + 2 \lambda \mu + \mu^2 - 4 \lambda \mu r} }{2 \lambda r} \]
obtained by setting 
\[\frac{d}{dt}H_2 = 0=\lambda r H_2^2 - (\lambda + \mu) H_2 + \mu .\]
Next,  perform a change of variables $z = \frac{1}{H_2 - y_b}$ so that $H_2 = y_1 + \frac{1}{z}$, and thus
\[ \frac{dz}{dt} + (2 y_b \lambda r - \lambda - \mu) z = -\lambda r \]
Using the multiplier method with multiplier $\exp{ \left\{ (2 y_b \lambda r  - \lambda - \mu) t \right\} }$, we obtain
\[ z = e^{-(2y_b \lambda r - \lambda - \mu)t} \left[ \int -\lambda r e^{(2\lambda r y_b - \lambda - \mu)t} dt + C \right]  = \frac{-\lambda r}{2 \lambda r y_b - \lambda - \mu} + C e^{-(2y_b \lambda r - \lambda - \mu)t}. \]
Thus,
\[H_2 = y_b + \frac{1}{z} = y_b + \left[  \frac{-\lambda r}{2 \lambda r y_b - \lambda - \mu} + C e^{-(2y_b \lambda r - \lambda - \mu)t} \right] ^{-1} \]
and from $H_2(0, r, s_1, s_2) = s_2$, we see $C = \frac{1}{s_2 - y_b} + \frac{\lambda r}{2 \lambda r y_b - \lambda - \mu}$.
Finally, we arrive at the full solution to the second ODE
\[ H_2 := g^b(t,s_1,s_2,r) = y_b + \left[  \frac{-\lambda r}{2 \lambda r y_b - \lambda - \mu} + \left( \frac{1}{s_2 - y_b} + \frac{\lambda r}{2 \lambda r y_b - \lambda - \mu} \right)  e^{-(2y_b \lambda r - \lambda - \mu)t} \right] ^{-1}. \]
Plugging this solution into the equation for $H_1$, we have a single ODE that is numerically solvable:
\[ \frac{d}{dt} H_1^+(t,s_1,s_2,r) = \lambda r H_1^+ g^b + \nu g^b + \mu - (\lambda + \mu + \nu) H_1^+ . \]
An analogous solution beginning with Equations \eqref{eq:deathPseudo},  \eqref{eq:shiftPseudo},  and \eqref{eq:PTsystem} instead of \eqref{eq:backwardBirths} and solving the second Ricatti equation is used to simplify the other equation systems, yielding the results presented in Theorem \ref{thm:main}.
\end{proof}

\section*{Appendix C}
Here we include additional figures that support, but are not crucial to, illustrating our simulation results.

Figure \ref{fig:shiftComparison} displays the transition probabilities $p_{(10,0),(ij)}$ for 25 randomly sampled  $(i,j)$ pairs with $0 \leq i,j \leq 32$, calculated by our generating function approach alongside their Monte Carlo estimates and confidence intervals. Monte Carlo estimates are based on $5000$ realizations beginning with an initial count of 10 with $dt=1.0, \lambda=.5, \mu~=.45$ and $\nu$ ranging from $0.3$ to $2.0$. 
\begin{figure}[H]
\centering
\includegraphics[width=.7\paperwidth]{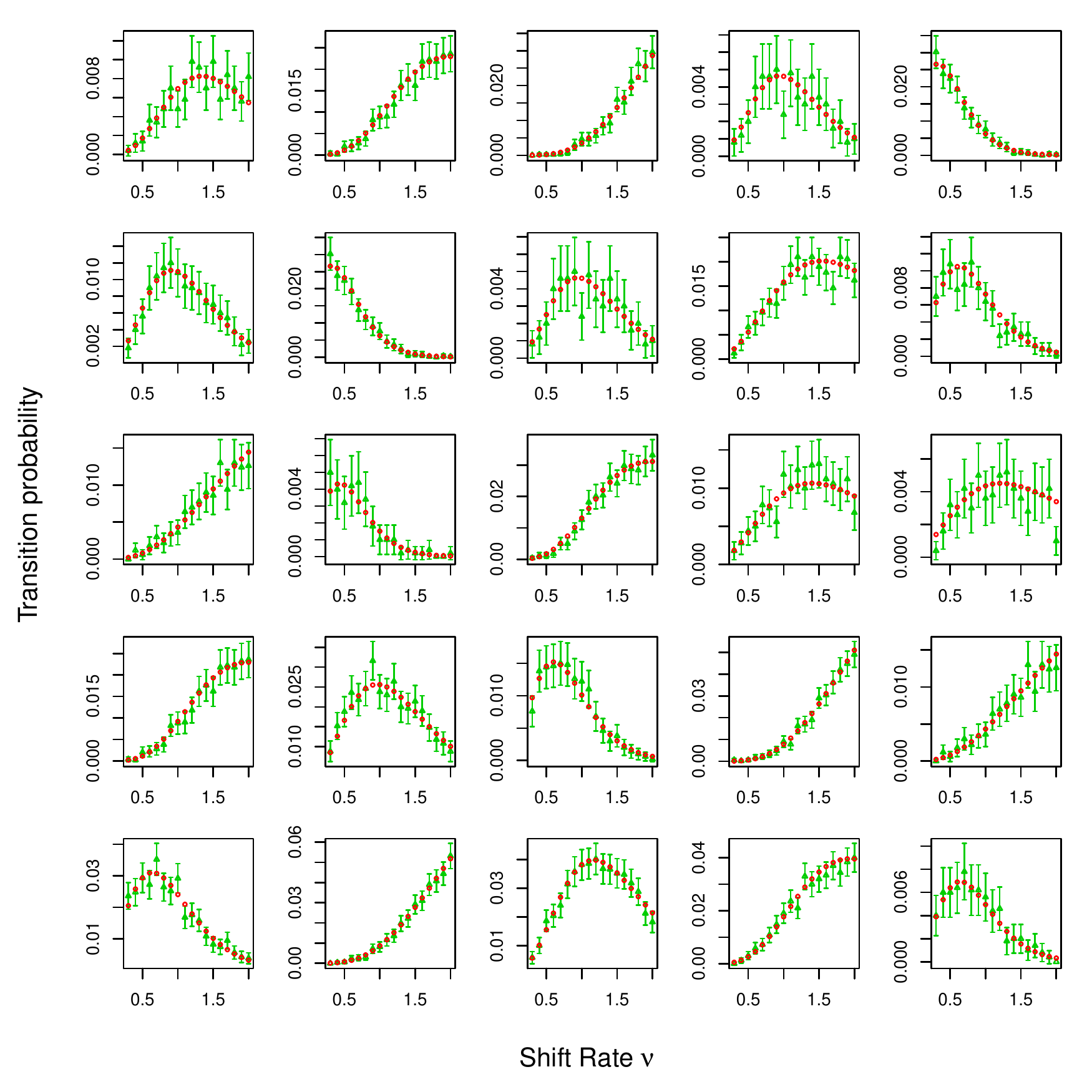}
\caption{Transition probabilities remain accurate when increasing all rates of process, presented over a wide range of $\nu$ values. Green points and intervals correspond to Monte Carlo estimates of transition probabilities and corresponding $95\%$ confidence intervals. The red 
points denote probabilities computed with our generating function method.}
\label{fig:shiftComparison}
\end{figure}

Figure \ref{fig:restricted} shows that restricted moment calculations performed during the E-step are indeed accurate: the following figure corresponds to simulations with 3 times the rates in the Rosenberg-Tanaka paper: $(\lambda, \nu, \mu) = 3 \cdot (.0188, .0026, .0147)$, with 10 initial particles and varying time lengths.
\begin{figure}[H]
\centering
\includegraphics[width = .65\paperwidth]{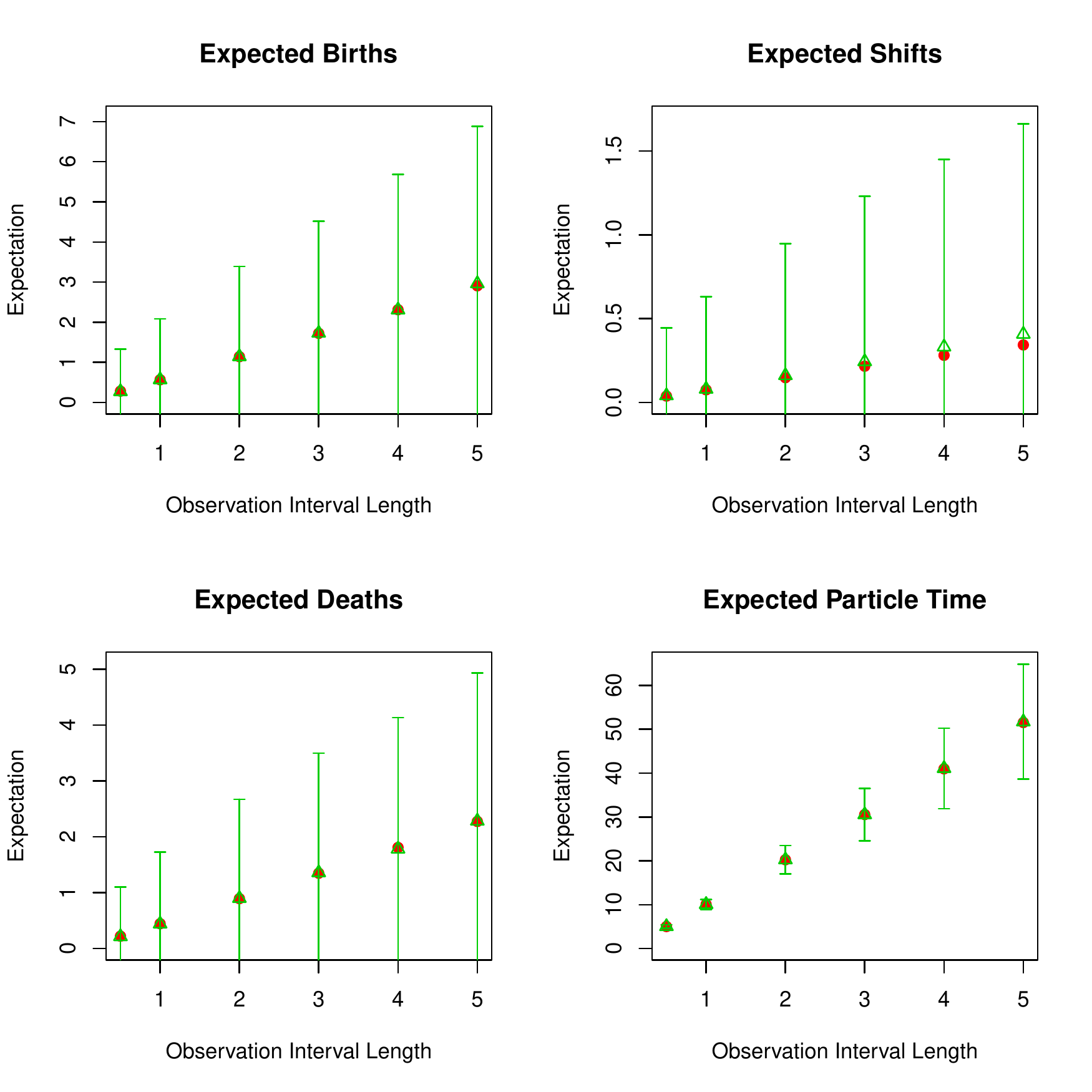}
\caption{Restricted moments calculated by our method (red) compared to approximation over 5000 Monte Carlo simulations and corresponding $95\%$ confidence intervals (green).}
\label{fig:restricted}
\end{figure}

\end{document}